\documentclass[letterpaper, 10 pt, conference]{ieeeconf}
\IEEEoverridecommandlockouts    
\usepackage{graphicx}
\usepackage{amsmath}
\usepackage{amssymb}
\usepackage{mathtools}

\usepackage{amsfonts}
\usepackage{caption}
\usepackage{subcaption}
\usepackage{color}

\newtheorem{problem}{Problem}
\newtheorem{definition}{Definition}
\newtheorem{assumption}{Assumption}
\newtheorem{theorem}{Theorem}
\newtheorem{lemma}{Lemma}
\newtheorem{remark}{Remark}

\begin{document}

\title{\LARGE \bf Distributed Density Filtering for Large-Scale Systems Using Mean-Filed Models}
\author{Tongjia Zheng and Hai Lin
\thanks{*This work was supported by the National Science Foundation under Grant No. IIS-1724070, CNS-1830335, IIS-2007949.}
\thanks{Tongia Zheng and Hai Lin are with the Department of Electrical Engineering, University of Notre Dame, Notre Dame, IN 46556, USA (e-mail: tzheng1@nd.edu, hlin1@nd.edu.). }
}

\maketitle

\thispagestyle{empty}
\pagestyle{empty}

\begin{abstract}
This work studies distributed (probability) density estimation of large-scale systems.
Such problems are motivated by many density-based distributed control tasks in which the real-time density of the swarm is used as feedback information, such as sensor deployment and city traffic scheduling.
This work is built upon our previous work \cite{zheng2020pde} which presented a (centralized) density filter to estimate the dynamic density of large-scale systems through a novel integration of mean-field models, kernel density estimation (KDE), and infinite-dimensional Kalman filters.
In this work, we further study how to decentralize the density filter such that each agent can estimate the global density only based on its local observation and communication with neighbors.
This is achieved by noting that the global observation constructed by KDE is an average of the local kernels.
Hence, dynamic average consensus algorithms are used for each agent to track the global observation in a distributed way.
We present a distributed density filter which requires very little information exchange, and study its stability and optimality using the notion of input-to-state stability.
Simulation results suggest that the distributed filter is able to converge to the centralized filter and remain close to it.
\end{abstract}


\section{Introduction}
In recent years, density-based optimization and control strategies for large-scale systems have becoming increasingly popular.
The general objective is to control/optimize (a functional of) the real-time distribution of the swarm \cite{ferrari2016distributed, elamvazhuthi2015optimal, zheng2020complex}, which has a variety of applications such as sensor deployment
and city traffic scheduling.
To ensure stability and robustness, the global density of the swarm is usually fed back into the algorithms to form a closed loop \cite{krishnan2018distributed, zheng2020transporting}.
Considering the scalability issue of the control algorithms and the privacy issue of data sharing, it is desirable to estimate the global density and implement the control strategy in a fully distributed manner using only local observations and information exchange.
In such scenarios, the agents (such as mobile sensors) are built by task designers, i.e., their dynamics are known.
This motivates the problem of how to take advantage of the available dynamics and estimate the global density of the swarm in a distributed manner.

Density estimation is a fundamental problem in statistics and has been studied using various methods, including parametric and nonparametric methods.
Parametric algorithms assume that the samples are drawn from a known parametric family of distributions and learn the parameters to maximize the likelihood \cite{silverman1986density}. 
Several distributed parametric techniques exist. 
For example, in \cite{gu2008distributed}, the unknown density is represented by a mixture of Gaussians, and the parameters are determined by a combination of expectation maximization (EM) algorithms and consensus protocols.
Performance of such estimators rely on the validity of the assumed models, and therefore they are unsuitable for estimating an evolving density. 
In non-parametric approaches, the data are allowed to speak for themselves in determining the density estimate, where \textit{kernel density estimation (KDE)} \cite{silverman1986density} is the most popular choice.
A distributed KDE algorithm is given in \cite{hu2007distributed}, which uses an information sharing protocol to incrementally exchange kernel information between sensors until a complete and accurate approximation of the global KDE is achieved by each sensor.
All these algorithms aim at estimating a static density.
To the best of our knowledge, distributed estimation for dynamic density remains largely unexplored.

When the dynamics of the agents/samples are available, an alternative way is to consider a filtering problem of estimating the distribution of all the agents' states.
There exists a large body of literature for the filtering problem, ranging from the celebrated Kalman filters and their variants \cite{julier2004unscented} to the more general Bayesian filters and particle filters \cite{chen2003bayesian}.
Motivated by the development of sensor networks, distributed implementation for these filters have also been extensively studied \cite{olfati2005distributed}.
The general strategy is that each agent runs a local filter based on its local information, and exchanges its information and/or estimate with neighboring agents to gradually estimate the global distribution.
However, stability analysis and implementation are known to be difficult when the agents' dynamics are nonlinear and time-varying.

In summary, considering the requirements of decentralization, convergence and efficiency, existing methods are unsuitable for estimating the time-varying density of large-scale systems in a distributed manner. 
This motivates us to propose a distributed, dynamic and scalable density estimation algorithm that can perform online and use only local observation and communication to guarantee its convergence. 
In our previous work \cite{zheng2020pde}, we proposed a (centralized) density filter through a novel integration of mean-field models, KDE and infinite-dimensional Kalman filters, which was proved to be convergent and efficient.
In this work, we decentralize the density filter by replacing the global information with a dynamic consensus protocol, and show that its performance converges to the centralized filter.
Our contribution is summarized as follows: (i) We present a distributed density filter such that each agent estimates the global density only through local observations and very little amount of communication; (ii) The distributed filter is proved to converge to the centralized filter in the sense of input-to-state stability;
(iii) All the results hold even if the agents' dynamics are nonlinear and time-varying.

The rest of the paper is organized as follows. 
Section \ref{section:preliminaries} introduces some preliminaries. 
Problem formulation is given in Section \ref{section:problem formulation}. 
Section \ref{section:review of centralized density filter} reviews the centralized filter and its basic property.
Section \ref{section:distributed density filter} is our main result which presents a distributed density filter and then studies its stability and optimality. 
Section \ref{section:simulation} performs an agent-based simulation to verify the effectiveness of the distributed filter. 
Section \ref{section:conclusion} summarizes the contribution and points out future research.

\section{Preliminaries} \label{section:preliminaries}
\subsection{Infinite-Dimensional Kalman Filters} \label{section:Kalman filter}
The Kalman filter is an algorithm that uses the system's model and sequential measurements to gradually improve its state estimate. 
Its extension to infinite-dimensional systems is studied in \cite{curtain1975infinite}. 
Formally, assume the signal $x(t)$ and its measurement $y(t)$, both in a Hilbert space, are generated by the stochastic linear differential equations:
\begin{align*}
    & dx=\mathcal{A}(t) x d t+\mathcal{B}(t)q(t) d v, \quad x(t_{0})=x_{0}, \quad E[x_{0}]=0,\\
    & dy=\mathcal{C}(t) x d t+ r(t)d w, \quad y(t_{0})=y_{0},
\end{align*}
where $dv$ and $dw$ are infinite-dimensional Wiener processes with incremental covariance operators $\mathcal{V}$ and $\mathcal{W}$ respectively. 
Assume $\operatorname{Cov}[v(t), w(\tau)]=0$ and $E[\langle v(t), w(\tau)\rangle]=0$ for $t\neq\tau$. 
Denote $\mathcal{Q}(t)=q(t)\mathcal{V}q^*(t)$ and $\mathcal{R}(t)=r(t)\mathcal{W}r^*(t)$. 
The infinite-dimensional Kalman filter is given by:
$$
d \hat{x}=\mathcal{A}(t) \hat{x} d t+\mathcal{L}(t)[y(t)-\mathcal{C}(t) \hat{x}] d t,
$$
where $\mathcal{L}(t)=\mathcal{P}(t) \mathcal{C}^{*}(t) \mathcal{R}^{-1}(t)$ is the optimal Kalman gain and $\mathcal{P}(t)$ is the solution of the operator Riccati equation
$$
\dot{\mathcal{P}}=\mathcal{A}(t) \mathcal{P}+\mathcal{P} \mathcal{A}^*(t)-\mathcal{P} \mathcal{C}^{*}(t) \mathcal{R}^{-1}(t) \mathcal{C}(t)\mathcal{P}+\mathcal{B}(t) \mathcal{Q}(t) \mathcal{B}^{*}(t)
$$
with $\mathcal{P}(t_{0})=\operatorname{Cov}[x_{0}, x_{0}]$. 

\subsection{Input-to-state stability}
Input-to-state stability (ISS) is a stability notion for studying nonlinear control systems with external inputs \cite{sontag1995characterizations}.
The extension for infinite-dimensional systems is studied in \cite{dashkovskiy2013input}.
We introduce the following comparison functions:
\begin{align*}
    \mathcal{K} &:= \{\gamma :\mathbb{R}_{+} \to \mathbb{R}_{+}|\gamma \text{ is continuous and strictly} \\
    &\quad\quad \text{increasing}, \gamma (0)=0\},\\
    \mathcal{L} &:= \{\gamma  : \mathbb{R}_{+} \to \mathbb{R}_{+} | \gamma \text{ is continuous and strictly } \\
    &\quad\quad \text{decreasing with } \lim_{t\to \infty} \gamma (t) = 0\},  \\
    \mathcal{KL} &:=\{\beta:\mathbb{R}_{+} \times\mathbb{R}_{+} \to\mathbb{R}_{+} |\beta \text{ is continuous}, \beta(\cdot,t)\in \mathcal{K}, \\
    &\quad\quad \beta(r,\cdot)\in \mathcal{L}, \forall t\geq 0, \forall r>0\}.
\end{align*}

\begin{definition}[ISS \cite{dashkovskiy2013input}]
\label{definition:(L)ISS}
Consider a control system $\Sigma = (X, U, \phi)$ consisting of normed linear spaces $(X, \|\cdot\|_X)$ and $(U, \|\cdot\|_U)$, called the state space and the input space, endowed with the norms $\|\cdot\|_X$ and $\|\cdot\|_U$ respectively, and a transition map $\phi:\mathbb{R}_{+}\times X \times U \to X$. The system is said to be ISS if there exist $\beta\in \mathcal{KL}$ and $\gamma\in \mathcal{K}$, such that
$$
\|\phi(t,x_0,u)\|_X \leq \beta(\|x_0\|_X,t-t_0) + \gamma(\sup_{t_0\leq\tau\leq t}\|u(\tau)\|_U)
$$
holds $\forall x_0\in X$, $\forall t\geq t_0$ and $\forall u\in U$. It is called locally input-to-state stable (LISS), if there also exists constants $\rho_{x}, \rho_{u}>0$ such that the above inequality holds $\forall x_0 :\|x_0\|_{X} \leq \rho_{x}, \forall t \geq t_0$ and $\forall u \in U :\|u\|_{U} \leq \rho_{u}$.
\end{definition}

\section{Problem formulation} \label{section:problem formulation} 
This paper studies the problem of estimating the dynamically varying probability density of large-scale stochastic agents. 
Their dynamics are assumed to be known and satisfy:
\begin{equation} \label{eq:Langevin equation}
    dX_t^i=\boldsymbol{v}(X_t^i, t) d t+\boldsymbol{\sigma}(X_t^i, t) d B_t^i, \quad i = 1,\dots,N,
\end{equation}
where $X_t^i\in\mathbb{R}^n$ is the state of the $i$-th agent, $\boldsymbol{v}=(v_{1}, \ldots, v_n)\in\mathbb{R}^n$ is the deterministic dynamics, $B_{t}^i$ is an $m$-dimensional standard Wiener process which is independent across the agents, and $\boldsymbol{\sigma}=[\sigma_{jk}]\in\mathbb{R}^{n\times m}$ is the stochastic dynamics. 
The states $\{X_t^i\}_{i=1}^{N}$ are assumed to be observable. 
The probability density $p(x, t)$ of the states is known to satisfy a mean-field partial differential equation, called the Fokker-Planck equation:
\begin{equation} \label{eq:FP equation}
\begin{array}{cl}
     &\displaystyle\frac{\partial p(x, t)}{\partial t}=-\sum_{l=1}^{n} \frac{\partial}{\partial x_{l}}[v_{l}(x, t)p(x, t)] \\
     & \quad\quad\quad\quad\quad \displaystyle+\sum_{l=1}^{n} \sum_{j=1}^{n} \frac{\partial^{2}}{\partial x_{l} \partial x_{j}}[D_{l j}(x, t) p(x, t)],\\
     &p(\cdot,t_0)=p_0,
\end{array}
\end{equation}
where $D_{lj}(x,t)=\frac{1}{2}\sum_{k=1}^{m}\sigma_{lk}(x,t)\sigma_{jk}(x,t)$, $x=(x_1,\dots,x_n)$, and $p_0$ is the initial density.
If the states are confined within a bounded domain $\Omega$, we can impose a reflecting boundary condition:
\begin{equation} \label{eq:zero-flux BC}
\mathbf{n} \cdot(\boldsymbol{g}-\boldsymbol{v}p)=0, \quad\text{on } \partial \Omega,
\end{equation}
where $\boldsymbol{g}=(\sum_{j=1}^{n}\frac{\partial}{\partial x_j}(D_{1j}p),\ldots,\sum_{j=1}^{n}\frac{\partial}{\partial x_j}(D_{nj}p))$, $\partial\Omega$ is the boundary of $\Omega$ and $\mathbf{n}$ is the outward normal to $\partial\Omega$. 

\begin{remark} \label{remark:nonlinear and time-varying}
Note that \eqref{eq:FP equation} is uniquely determined by \eqref{eq:Langevin equation} which is known. 
This relationship holds even if \eqref{eq:Langevin equation} is nonlinear and time-varying.
Hence, the centralized/distributed density filter we design apply to a large family of systems. 
Also note that \eqref{eq:FP equation} is always linear even if \eqref{eq:Langevin equation} is not. 
\end{remark}

We assume the agents can exchange information with neighbors to form a \textit{time-varying} topology $G(t)=(V,E(t))$, where $V$ is the set of $N$ nodes and $E(t)\subset V\times V$ is the set of communication links.
We assume $G(t)$ is undirected.
Now, we can formally state the problem to be solved as follows:

\begin{problem}[Distributed density estimation]
\label{problem:distributed density estimation}
Given the system \eqref{eq:FP equation}, the state $X_t^i$ of an agent $i$, and the topology $G$, we want to design communication and estimation protocols for each agent to estimate the global density $p(x,t)$.
\end{problem}

\section{Review of centralized density filters} \label{section:review of centralized density filter}
In this section, we review a centralized density filter and its stability/optimality property presented in our previous work \cite{zheng2020pde}.
The centralized density estimation problem assumes full availability of the agents' states and is stated as follows:

\begin{problem}[Centralized density estimation] 
Given the system \eqref{eq:FP equation} and agent states $\{X_t^i\}_{i=1}^{N}$, we want to estimate their density $p(x,t)$.
\end{problem}



The centralized density filter combines mean-field models, KDE and infinite-dimensional Kalman filters to gradually improve its estimate of the state of \eqref{eq:FP equation}.
Given $\{X^i\}_{i=1}^{N}\subseteq\mathbb{R}^n$, the kernel density estimator is given by \cite{silverman1986density}:
\begin{equation} \label{eq:KDE}
f_N(x) = \frac{1}{N h^{n}} \sum_{i=1}^{N} K\Big(\frac{1}{h}(x-X^{i})\Big),
\end{equation}
where $K$ is a kernel function and $h$ is the bandwidth.
Under proper choice of $h$, $f_N(x)$ is asymptotically normal, and $f_N(x_i)$ and $f_N(x_j)$ are asymptotically uncorrelated for $x_i\neq x_j$, as $N\to\infty$ \cite{cacoullos1966estimation}. 
Hence, $f_N(x)-f(x)$ is approximately Gaussian with independent components when $N$ is large.

To design a filter, we rewrite \eqref{eq:FP equation} as an evolution equation and use KDE to construct a noisy measurement $y(t)$:
\begin{equation} \label{eq:evolution equation of density}
\begin{aligned}
    \Dot{p}(t) &= \mathcal{A}(t)p(t),
    \\
    y(t)&= p_{\text{KDE}}(t) = p(t)+w(t),
\end{aligned}
\end{equation}
where $\mathcal{A}(t)=-\sum_{l=1}^{n} \frac{\partial}{\partial x_{l}}(v_{l}\cdot)+\sum_{l=1}^{n} \sum_{j=1}^{n} \frac{\partial^{2}}{\partial x_{l} \partial x_{j}}(D_{l j}\cdot)$ is a linear operator, $p_{\text{KDE}}(t)$ represents a kernel density estimator using $\{X_t^i\}_{i=1}^{N}$, and $w(t)$ is the measurement noise which is approximately Gaussian with covariance $\mathcal{R}(t)=k\operatorname{diag}(p(t))$ where $k>0$ is a constant depending on $K$, $N$ and $h$.
The optimal density filter is given
\begin{equation} \label{eq:optimal density filter}
    \Dot{\Hat{p}} = \mathcal{A}(t)\Hat{p}+\mathcal{L}(t)(y-\Hat{p}),\quad \Hat{p}(t_0)=p_{\text{KDE}}(t_0),
\end{equation}
where $\mathcal{L}(t)=\mathcal{\mathcal{P}}(t) \mathcal{R}^{-1}(t)$ is the optimal Kalman gain and $\mathcal{P}(t)$ is a solution of the following operator Riccati equation 
\begin{equation} \label{eq:optimal Riccati}
    \dot{\mathcal{P}}=\mathcal{A}(t)\mathcal{P}+\mathcal{P}\mathcal{A}^*(t)-\mathcal{P}\mathcal{R}^{-1}(t)\mathcal{P}.
\end{equation}
Since $\mathcal{R}(t)$ depends on the unknown density $p(t)$, we approximate $\mathcal{R}(t)$ using $\Bar{\mathcal{R}}(t)=\Bar{k}\operatorname{diag}(p_{\text{KDE}}(t))$, where $\Bar{k}$ is computed as $\Bar{k}=(\int[K(u)]^2du)/(Nh^n)$ \cite{zheng2020pde}. 
Correspondingly, we obtain a ``suboptimal'' density filter:
\begin{equation} \label{eq:suboptimal density filter}
    \Dot{\Hat{p}} = \mathcal{A}(t)\Hat{p}+\Bar{\mathcal{L}}(t)(y-\Hat{p}),\quad \Hat{p}(t_0)=p_{\text{KDE}}(t_0)
\end{equation}
where $\Bar{\mathcal{L}}(t)=\Bar{\mathcal{P}}(t)\Bar{\mathcal{R}}^{-1}(t)$ is the suboptimal Kalman gain and $\Bar{\mathcal{P}}(t)$ is a solution of the approximated Riccati equation
\begin{equation} \label{eq:suboptimal Riccati}
    \dot{\Bar{\mathcal{P}}}=\mathcal{A}(t)\Bar{\mathcal{P}}+\Bar{\mathcal{P}}\mathcal{A}^*(t)-\Bar{\mathcal{P}}\Bar{\mathcal{R}}^{-1}(t)\Bar{\mathcal{P}}.
\end{equation}
To study the stability of the suboptimal filter, define $\Tilde{p} = \Hat{p}-p$. 
Then along $\Bar{{\mathcal{P}}}(t)$ we have
\begin{equation} \label{eq:estimation error equation}
    \Dot{\Tilde{p}} = (\mathcal{A}(t)-\Bar{{\mathcal{P}}}\Bar{\mathcal{R}}^{-1}(t))\Tilde{p}+\Bar{{\mathcal{P}}}\Bar{\mathcal{R}}^{-1}(t)w.
\end{equation}
Define $\Gamma=\Bar{{\mathcal{P}}}-{\mathcal{P}}$. 
Using \eqref{eq:optimal Riccati} and \eqref{eq:suboptimal Riccati} we have
\begin{equation} \label{eq:Riccati error equation1}
\begin{aligned}
    &\Dot{\Gamma} = \mathcal{A}(t)\Gamma+\Gamma \mathcal{A}^*(t) - \Bar{{\mathcal{P}}}\Bar{\mathcal{R}}^{-1}(t)\Bar{{\mathcal{P}}} + {\mathcal{P}} \mathcal{R}^{-1}(t){\mathcal{P}}.
\end{aligned}
\end{equation}

In \cite{zheng2020pde}, we have proved that (under mild conditions): (i) the estimation error \eqref{eq:estimation error equation} is stable; (ii) the solution of \eqref{eq:suboptimal Riccati} remains close to the solution of \eqref{eq:optimal Riccati}; and (iii) the suboptimal gain $\Bar{\mathcal{L}}$ remains close to the optimal gain $\mathcal{L}$.
Formally, we define $\|\Bar{\mathcal{R}}^{-1}-\mathcal{R}^{-1}\|$ as the approximation error, as it is zero if and only if $\Bar{\mathcal{R}}=\mathcal{R}$.
The stability results are formally stated in the following theorem, whose proof can be found in \cite{zheng2020pde}.

\begin{theorem} \label{thm:stability of centralized density filter}
\cite{zheng2020pde} Assume that $\|{\mathcal{P}}(t)\|$ and $\|\Bar{{\mathcal{P}}}(t)\|$ are uniformly bounded, and that there exist positive constants $c_1$ and $c_2$ such that for all $t\geq t_0$,
\begin{equation}\label{eq:uniform positivity}
    0<c_1I\leq \mathcal{R}^{-1}(t),\Bar{\mathcal{R}}^{-1}(t),{\mathcal{P}}^{-1}(t),\Bar{{\mathcal{P}}}^{-1}(t)\leq c_2I.
\end{equation}
Then we have the following conclusions:
\begin{itemize}
    \item[(i)] The noise-free part (i.e., $w=0$) of \eqref{eq:estimation error equation} is uniformly exponentially stable.
    \item[(ii)] System \eqref{eq:Riccati error equation1} is LISS with respect to $\|\Bar{\mathcal{R}}^{-1}-\mathcal{R}^{-1}\|$.
    \item[(iii)] $\|\Bar{\mathcal{L}}-\mathcal{L}\|$ is LISS with respect to $\|\Bar{\mathcal{R}}^{-1}-\mathcal{R}^{-1}\|$.
\end{itemize}
\end{theorem}

\begin{remark}
A few comments are in order.
By taking advantage of the dynamics, this density filter essentially combines past outputs to produce better and convergent estimates.
It is scalable because we lift the density estimation problem from a very large finite-dimensional space (of the agents' states) into an infinite-dimensional space (of densities) by using mean-field models.
The performance becomes better when the agents' population is larger.
It is computationally efficient and can be computed online because the involved matrices in its numerical implementation are highly sparse.
\end{remark}

\section{Distributed density estimation}\label{section:distributed density filter}
In this section, we present a distributed density filter by integrating consensus protocols into \eqref{eq:suboptimal density filter} and \eqref{eq:suboptimal Riccati}, and study its convergence and optimality.

We reformulate the system \eqref{eq:evolution equation of density} in the distributed form:
\begin{equation} \label{eq:distributed evolution equation of density}
\begin{aligned}
    &\Dot{p}(t) = \mathcal{A}(t)p(t),
    \\
    &z_i(t) =K_i(t), \quad i=1,\dots,N,
\end{aligned}
\end{equation}
where $K_i(t)=\frac{1}{h^n}K\big(\frac{1}{h}(x-X_t^i)\big)$ is a kernel centered at position $X_t^i$.
We view $z_i(t)$ as the local measurement made by the $i$-th agent.
We may write $z_i=r_i+w$, where $w$ is the noise defined in \eqref{eq:evolution equation of density} and $r_i:=z_i-w$ is the deterministic component of $z_i$ with $\frac{1}{N}\sum_{i=1}^Nr_i=p$.

The challenge of distributed density estimation lies in that each agent alone does not have any meaningful observation of the unknown density $p(t)$, because its local measurement $z_i$ is simply a kernel centered at position $X_t^i$, which conveys no information about $p(t)$.


We design the local density filter for each agent $i$ as
\begin{equation} \label{eq:Type-I density filter}
    \Dot{\Hat{p}}_i = \mathcal{A}(t)\Hat{p}_i+\Bar{\mathcal{L}}_i(t)(y_i-\Hat{p}_i),\quad \Hat{p}_i(t_0)=y_i(t_0),
\end{equation}
where $y_i(t)$ is to be constructed later, $\Bar{\mathcal{L}}_i(t)=\Bar{\mathcal{P}}_i(t)\Bar{\mathcal{R}}_i^{-1}(t)$ is the local Kalman gain with $\Bar{\mathcal{R}}_i(t)=\Bar{k}\operatorname{diag}(y_i(t))$, and $\Bar{\mathcal{P}}_i(t)$ is a solution of the local operator Riccati equation 
\begin{equation} \label{eq:local suboptimal Riccati 1}
    \dot{\Bar{\mathcal{P}}}_i=\mathcal{A}(t)\Bar{\mathcal{P}}_i+\Bar{\mathcal{P}}_i\mathcal{A}^*(t)-\Bar{\mathcal{P}}_i\Bar{\mathcal{R}}_i^{-1}(t)\Bar{\mathcal{P}}_i.
\end{equation}

An important observation is that $y(t)=\frac{1}{N}\sum_i^Nz_i(t)$ and $\Bar{\mathcal{R}}(t)=\operatorname{diag}(\frac{1}{N}\sum_i^Nz_i(t))=\frac{1}{N}\sum_i^N\operatorname{diag}(z_i(t))$, which suggests that if we can design algorithms such that $y_i\to\frac{1}{N}\sum_i^Nz_i(t)$ and $\Bar{\mathcal{R}}_i\to\frac{1}{N}\sum_i^N\operatorname{diag}(z_i(t))$, then the local filter \eqref{eq:Type-I density filter} and \eqref{eq:local suboptimal Riccati 1} should converge to the centralized filter \eqref{eq:suboptimal density filter} and \eqref{eq:suboptimal Riccati}.
(Note that it is sufficient to only study $y_i\to\frac{1}{N}\sum_i^Nz_i(t)$.)
This property is first observed in \cite{olfati2005distributed}.
The associated problem of tracking the average of $N$ time-varying reference signals is called dynamic average consensus \cite{kia2019tutorial}.
We adopt the proportional-integral (PI) consensus estimator given in \cite{freeman2006stability} which is a low-pass filter.
Some useful properties are summarized in Appendices.
According to \eqref{eq:PI consensus estimator}, we construct $y_i$ in the following way:
\begin{align}
    \partial_ty_i &=-\alpha\left(y_{i}-z_{i}\right)-\sum_{j=1}^{N} a_{i j}\left(y_{i}-y_{j}\right)+\sum_{j=1}^{N} b_{j i}\left(\phi_{i}-\phi_{j}\right) \nonumber\\
    \partial_t\phi_i &=-\sum_{j=1}^{N} b_{i j}\left(y_{i}-y_{j}\right),\label{eq:PI consensus estimator for kernels}
\end{align}
where the coefficients are described in \eqref{eq:PI consensus estimator}.
In other words, the consensus algorithm is performed pointwise.

Like Theorem \ref{thm:stability of centralized density filter}, we need the following mild assumption.

\begin{assumption} \label{assumption:uniform boundedness for local filter}
Assume $\|\Bar{{\mathcal{P}}}(t)\|$ and $\|\Bar{{\mathcal{P}}}_i(t)\|$ are uniformly bounded and there exist constants $c_3,c_4>0$ such that,
\begin{equation*}
    0<c_3I\leq \Bar{\mathcal{R}}^{-1}(t),\Bar{\mathcal{R}}_i^{-1}(t),\Bar{{\mathcal{P}}}^{-1}(t),\Bar{{\mathcal{P}}}_i^{-1}(t)\leq c_4I, \forall t.
\end{equation*}
\end{assumption}


\begin{theorem}\label{thm:ISS of PI consensus estimator for kernels}
Let $[a_{i j}(t)]$ and $[b_{i j}(t)]$ be such that the corresponding $L_{P}(t)$ and $L_{I}(t)$ are Laplacian matrices of strongly connected and weight-balanced digraphs. 
Then $\|y_i(t)-y(t)\|$ is ISS.
Moreover, under Assumption \ref{assumption:uniform boundedness for local filter}, $\|\Bar{\mathcal{R}}_i(t)-\Bar{\mathcal{R}}(t)\|$ and $\|\Bar{\mathcal{R}}_i^{-1}(t)-\Bar{\mathcal{R}}^{-1}(t)\|$ are both ISS.
\end{theorem}

\begin{proof}
This is a consequence of the ISS property of PI consensus estimators; see Appendices.
\end{proof}

\begin{remark}
In practice, the network may not be always strongly connected since the agents are mobile.
Nevertheless, the transient error caused by agents' permanent dropout will be slowly forgotten according to \eqref{eq:robust to initialization}.
\end{remark}


The remaining task is to show that \eqref{eq:Type-I density filter} and \eqref{eq:local suboptimal Riccati 1} indeed remain close to \eqref{eq:suboptimal density filter} and \eqref{eq:suboptimal Riccati}, respectively.
Towards this end, we define $\Gamma_i=\Bar{\mathcal{P}}-\Bar{\mathcal{P}}_i$. 
Using \eqref{eq:suboptimal Riccati} and \eqref{eq:local suboptimal Riccati 1} we have
\begin{equation} \label{eq:local Riccati error equation 1}
\begin{aligned}
    &\Dot{\Gamma}_i = \mathcal{A}(t)\Gamma_i+\Gamma_i \mathcal{A}^*(t) - \Bar{\mathcal{P}}\Bar{\mathcal{R}}^{-1}(t)\Bar{\mathcal{P}} + \Bar{\mathcal{P}}_i\Bar{\mathcal{R}}_i^{-1}(t)\Bar{\mathcal{P}}_i.
\end{aligned}
\end{equation}
Define $\Tilde{p}_i = \Hat{p}_i-p$. 
Then along $\Bar{\mathcal{P}}_i(t)$ we have
\begin{equation} \label{eq:local estimation error equation 1}
    \Dot{\Tilde{p}}_i = \mathcal{A}(t)\Tilde{p}_i + \Bar{\mathcal{P}}_i\Bar{\mathcal{R}}_i^{-1}(t)(y_i-\Hat{p}_i).
\end{equation}
The stability and optimality results are given as follows.

\begin{theorem} \label{thm:stability of Type-I density filter}
Assume each agent uses \eqref{eq:PI consensus estimator for kernels} to construct $y_i$ from $\{z_i\}_{i=1}^N$.
Let $[a_{i j}(t)]$ and $[b_{i j}(t)]$ satisfy the assumptions in Theorem \ref{thm:ISS of PI consensus estimator for kernels}.
Under Assumption \ref{assumption:uniform boundedness for local filter}, we have:
\begin{itemize}
    \item[(i)] The noise-free part (i.e., $w=0$) of \eqref{eq:local estimation error equation 1} is LISS.
    \item[(ii)] The system \eqref{eq:local Riccati error equation 1} is LISS.
    \item [(iii)] $\|\Bar{\mathcal{L}}_i-\Bar{\mathcal{L}}\|$ is LISS.
\end{itemize}
\end{theorem}

\begin{proof}
(i) Note that
\begin{equation} \label{eq:local estimation error equation 1_1}
\begin{split}
    \Dot{\Tilde{p}}_i &= \mathcal{A}(t)\Tilde{p}_i + \Bar{{\mathcal{P}}}_i\Bar{\mathcal{R}}_i^{-1}(t)(y_i-\Hat{p}_i) \\
    &=(\mathcal{A}(t)-\Bar{{\mathcal{P}}}_i\Bar{\mathcal{R}}_i^{-1}(t))\Tilde{p}_i + \Bar{{\mathcal{P}}}_i\Bar{\mathcal{R}}_i^{-1}(t)(y_i-p),
\end{split}
\end{equation}
which is a linear system.
Assuming $w=0$, we have $y_i-p=y_i-y$.
We first show that the unforced part of \eqref{eq:local estimation error equation 1_1}, i.e.,
\begin{equation}\label{eq:unforced estimation error equation1}
    \Dot{\Tilde{p}}_i = (\mathcal{A}(t)-\Bar{{\mathcal{P}}}_i\Bar{\mathcal{R}}_i^{-1}(t))\Tilde{p}_i,
\end{equation}
is uniformly exponentially stable.
To prove that, consider a Lyapunov functional $V_1 = \langle \Bar{{\mathcal{P}}}_i^{-1}\Tilde{p}_i,\Tilde{p}_i \rangle$.
We have
\begin{align*}
    \Dot{V}_1 &= \big\langle\Bar{{\mathcal{P}}}_i^{-1}\Dot{\Tilde{p}}_i,\Tilde{p}_i\big\rangle + \big\langle\Bar{{\mathcal{P}}}_i^{-1}\Tilde{p}_i,\Dot{\Tilde{p}}_i\big\rangle - \big\langle\Bar{{\mathcal{P}}}_i^{-1}\Dot{\Bar{{\mathcal{P}}}}_i\Bar{{\mathcal{P}}}_i^{-1}\Tilde{p}_i,\Tilde{p}_i\big\rangle\\
    &= \big\langle\Bar{{\mathcal{P}}}_i^{-1}(\mathcal{A}-\Bar{{\mathcal{P}}}_i\Bar{\mathcal{R}}_i^{-1})\Tilde{p}_i,\Tilde{p}_i\big\rangle + \big\langle(\mathcal{A}^*-\Bar{\mathcal{R}}_i^{-1}\Bar{{\mathcal{P}}}_i)\Bar{{\mathcal{P}}}_i^{-1}\Tilde{p}_i,\Tilde{p}_i\big\rangle\\
    & \quad -\big\langle \Bar{{\mathcal{P}}}_i^{-1}(\mathcal{A}\Bar{{\mathcal{P}}}_i+\Bar{{\mathcal{P}}}_i\mathcal{A}^{*}-\Bar{{\mathcal{P}}}_i\Bar{\mathcal{R}}_i^{-1}\Bar{{\mathcal{P}}}_i)\Bar{{\mathcal{P}}}_i^{-1}\Tilde{p}_i,\Tilde{p}_i \big\rangle\\
    &= -\langle \Bar{\mathcal{R}}_i^{-1}\Tilde{p}_i,\Tilde{p}_i \rangle,
\end{align*}
which shows that \eqref{eq:unforced estimation error equation1} is uniformly exponentially stable. 
Hence, for \eqref{eq:local estimation error equation 1_1}, there exist constants $\lambda_1,\theta_1>0$ such that
\begin{align*}
    \|\Tilde{p}_i\| &\leq e^{-\lambda_1(t-t_0)}\|\Tilde{p}_i(t_0)\| +\int_{t_0}^{t}e^{-\lambda_1(t-\tau)}\theta_1\|y_i(\tau)-y(\tau)\|d\tau \\
    &\leq e^{-\lambda_1(t-t_0)}\|\Tilde{p}_i(t_0)\|+\frac{\theta_1}{\lambda_1}\sup_{t_0\leq\tau\leq t}\|y_i(\tau)-y(\tau)\| \\
    &=:\beta_1(\|\Tilde{p}_i(t_0)\|,t-t_0) + \gamma_1(\sup_{t_0\leq\tau\leq t}\|y_i(\tau)-y(\tau)\|),
\end{align*}
where $\beta_1\in\mathcal{KL}$ and $\gamma_1\in\mathcal{K}$. 
Hence, \eqref{eq:local estimation error equation 1_1} is LISS to $\|y_i(\tau)-y(\tau)\|$.
Since $\|y_i(\tau)-y(\tau)\|$ itself is also ISS, the first statement results from that the cascade system of an ISS system and an LISS system is LISS \cite{khalil2002nonlinear}.

(ii) Using a Lyapunov functional $V_2=\langle\Bar{{\mathcal{P}}}^{-1}\Tilde{p},\Tilde{p}\rangle$, one can show that the following system along $\Bar{{\mathcal{P}}}(t)$ is also uniformly exponentially stable:
\begin{equation} \label{eq:unforced estimation error equation2}
    \Dot{\Tilde{p}} = (\mathcal{A}(t)-\Bar{{\mathcal{P}}}\Bar{\mathcal{R}}^{-1}(t))\Tilde{p}.
\end{equation}
Now rewrite \eqref{eq:local Riccati error equation 1} as
\begin{equation}\label{eq:local Riccati error equation 1_1}
\begin{aligned}
    \Dot{\Gamma}_i
    &=\mathcal{A}\Gamma_i+\Gamma_i \mathcal{A}^{*} - \Bar{{\mathcal{P}}}\Bar{\mathcal{R}}^{-1}\Gamma_i - \Bar{{\mathcal{P}}} \Bar{\mathcal{R}}^{-1}\Bar{{\mathcal{P}}}_i\\
    &\quad - \Gamma_i \Bar{\mathcal{R}}_i^{-1}\Bar{{\mathcal{P}}}_i + \Bar{{\mathcal{P}}}\Bar{\mathcal{R}}_i^{-1}\Bar{{\mathcal{P}}}_i\\
    &=(\mathcal{A}-\Bar{{\mathcal{P}}}\Bar{\mathcal{R}}^{-1})\Gamma_i + \Gamma_i(\mathcal{A}^{*}-\Bar{\mathcal{R}}_i^{-1}\Bar{{\mathcal{P}}}_i)\\
    &\quad - \Bar{{\mathcal{P}}}(\Bar{\mathcal{R}}^{-1}-\Bar{\mathcal{R}}_i^{-1})\Bar{{\mathcal{P}}}_i,
\end{aligned}
\end{equation}
which is a linear equation. 
Fix $q$ with $\|q\|=1$. 
Since \eqref{eq:unforced estimation error equation1} and \eqref{eq:unforced estimation error equation2} are uniformly exponentially stable, and $\|\Bar{{\mathcal{P}}}\|$ and $\|\Bar{{\mathcal{P}}}_i\|$ are uniformly bounded, there exist constants $\lambda_2,\theta_2>0$ s.t.
\begin{align*}
    \|\Gamma_i(t)q\|&\leq e^{-\lambda_2(t-t_0)}\|\Gamma_i(t_0)q\| \\
    &\quad+\int_{t_0}^{t}e^{-\lambda_2(t-\tau)}\theta_2\|\mathcal{R}^{-1}(\tau)-\mathcal{R}_i^{-1}(\tau)\|\|q\|d\tau.
\end{align*}
Similar to (i), we can prove that \eqref{eq:Riccati error equation1} is LISS with respect to $\|\Bar{\mathcal{R}}^{-1}-\bar{\mathcal{R}}_i^{-1}\|$.
Since $\|\Bar{\mathcal{R}}^{-1}-\mathcal{R}_i^{-1}\|$ itself is also LISS, the cascade system is LISS \cite{khalil2002nonlinear}.

(iii) To prove the third statement, observe that
\begin{align*}
    \|\Bar{\mathcal{L}}-\Bar{\mathcal{L}}_i\| &=\|\Bar{{\mathcal{P}}}\Bar{\mathcal{R}}^{-1}-\Bar{{\mathcal{P}}}_i \Bar{\mathcal{R}}_i^{-1}\| \\
    &\leq \|\Bar{{\mathcal{P}}}\Bar{\mathcal{R}}^{-1}-\Bar{{\mathcal{P}}}_i\Bar{\mathcal{R}}^{-1} + \Bar{{\mathcal{P}}}_i\Bar{\mathcal{R}}^{-1}-\Bar{{\mathcal{P}}}_i\Bar{\mathcal{R}}_i^{-1}\| \\
    &\leq \|\Bar{\mathcal{R}}^{-1}\|\|\Bar{{\mathcal{P}}}-\Bar{{\mathcal{P}}}_i\| + \|\Bar{{\mathcal{P}}}_i\|\|\Bar{\mathcal{R}}^{-1}-\Bar{\mathcal{R}}_i^{-1}\|.
\end{align*}
Hence, $\|\Bar{\mathcal{L}}-\Bar{\mathcal{L}}_i\|$ is LISS with respect to $\|\Bar{\mathcal{R}}^{-1}-\Bar{\mathcal{R}}_i^{-1}\|$ and the cascade system is also LISS \cite{khalil2002nonlinear}.
\end{proof}



 

\begin{remark}
This theorem states that the local filter \eqref{eq:Type-I density filter} remains close to the centralized filter \eqref{eq:suboptimal density filter}, and has comparable performance if the agents' states are slowly varying.
The performance also depends on the connectivity and switching rate of $G(t)$ \cite{freeman2006stability}.
Note that each local kernel $K_i(t)$ is uniquely determined by its center $X_t^i$.
Hence, to implement the local filter, each agent only needs to exchange its position with neighbors, which is very efficient.
\end{remark}

\section{Simulation studies}\label{section:simulation}
In this section, we study the performance of the distributed filter. 
We simulate $100$ agents within $\Omega=[0,1]^2$:
\begin{equation}
    dX_t^i=\nabla\cdot\frac{D\nabla f(x,t)}{f(x,t)} dt+DdB_t^i, \quad i = 1,\dots,100,
\end{equation}
where $D=0.03$ and $f(x,t)$ is a time-varying pdf to be specified. 
The initial positions $\{X_{t_0}^i\}_{i=1}^{100}$ are uniformly distributed. 
Therefore, the ground truth density satisfies
\begin{align} \label{eq:example of FP equation}
\begin{split}
     &\displaystyle\partial_t p(x,t) =-\nabla\cdot\frac{Dp(x,t)\nabla f(x,t)}{f(x,t)} + \frac{1}{2}D^2\Delta p(x,t), \\
     &p(\cdot,t_0)=1,
\end{split}
\end{align}
with a reflecting boundary condition \eqref{eq:zero-flux BC}. 
We design $f(x,t)$ as a mixture of two Gaussian components with common covariance $\text{diag}(0.02,0.02)$ and different time-varying means $[0.5+0.3\cos(0.04t),0.5+0.3\sin(0.04t)]^\intercal$ and $[0.5+0.3\cos(0.04t+\pi),0.5+0.3\sin(0.04t+\pi)]^\intercal$. 
With this design, the agents are nonlinear and time-varying and their states will concentrate to two ``spinning'' Gaussian components.
 
We use finite difference to numerically solve \eqref{eq:Type-I density filter} and \eqref{eq:local suboptimal Riccati 1}. 
Specifically, partition $\Omega$ into a $30\times30$ grid.
$\Hat{p}_i$ is represented as a $900\times1$ vector. 
$\mathcal{A}$, $\Bar{\mathcal{P}}_i$ and $\Bar{\mathcal{R}}_i$ are represented as $900\times900$ matrices. 
The initial condition for \eqref{eq:Type-I density filter} is chosen to be a ``flat'' Gaussian centered at $X_{t_0}^i$.
The time period for updating the filter and the consensus estimator are both $dt=0.1s$.
Note that $\mathcal{A}$ is highly sparse and $\Bar{\mathcal{R}}_i$ is diagonal, so the computation is very fast in general. 
For the PI estimator, we set $\alpha=0.2$, $a_{ij}=0.4$ and $b_{ij}=0.04$.
The communication distance of each agent is set to be $d=0.25$.

Simulation results are given in Fig. \ref{fig:density filter} where each column stands for a single time instance. 
As already observed in \cite{zheng2020pde}, the centralized filter quickly catches up with the evolution of the ground truth density and outperforms KDE.
We then randomly select an agent and investigate its local filter.
We observe that the local filter also gradually catches up with the ground truth density, but in a slower rate due to the delay effect of the dynamic consensus process.
In Fig. \ref{fig:error}, we compare the $L^2$ norms of estimation errors of KDE, the centralized filter, and eight randomly selected local filters, which verifies the convergence of the proposed local filter.

\begin{figure}[hbt!]
\setlength{\belowcaptionskip}{-0.4cm}
    \centering
    \includegraphics[width=0.7\columnwidth]{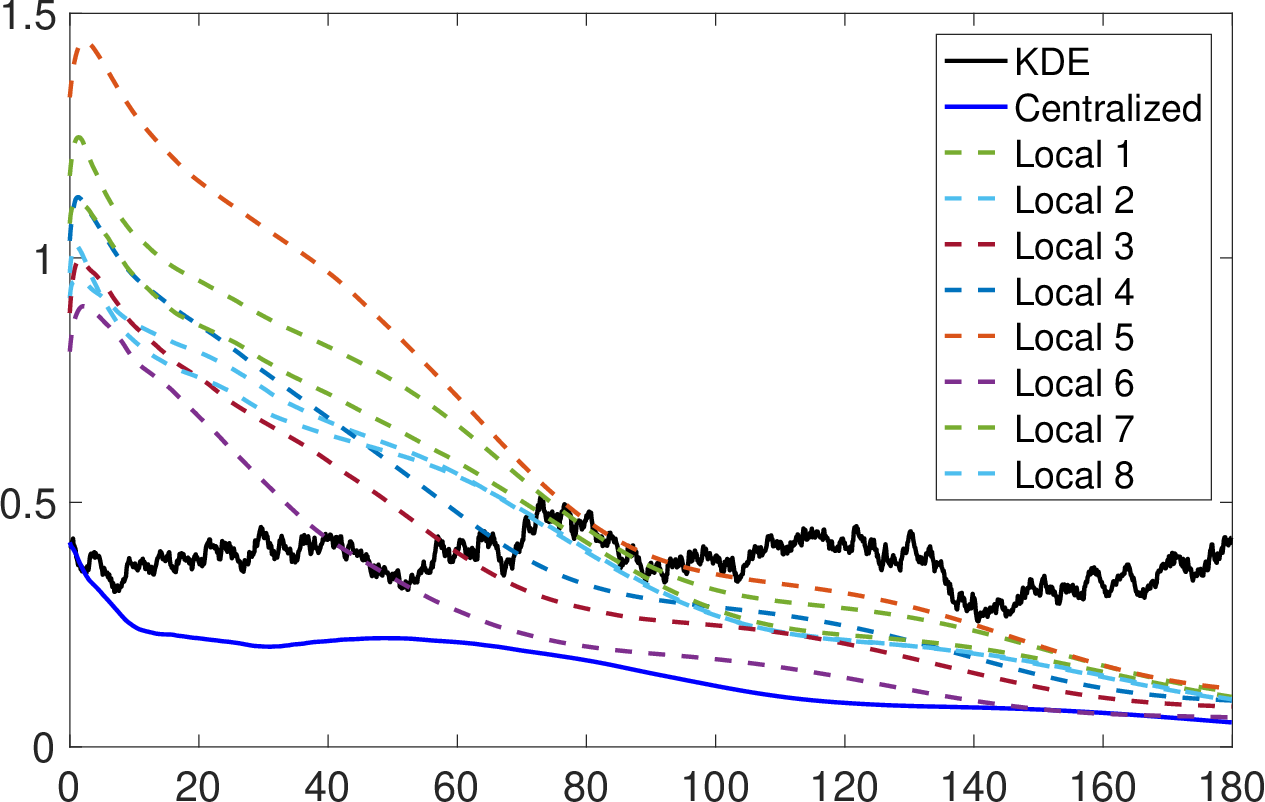}
    \caption{Estimation errors of KDE, the centralized filter, and eight randomly selected local filters.}
    \label{fig:error}
\end{figure}

\begin{figure*}[t]
\setlength{\belowcaptionskip}{-0.5cm}
    \centering
    \begin{subfigure}[b]{0.22\textwidth}
        \centering
        \includegraphics[width=0.95\textwidth]{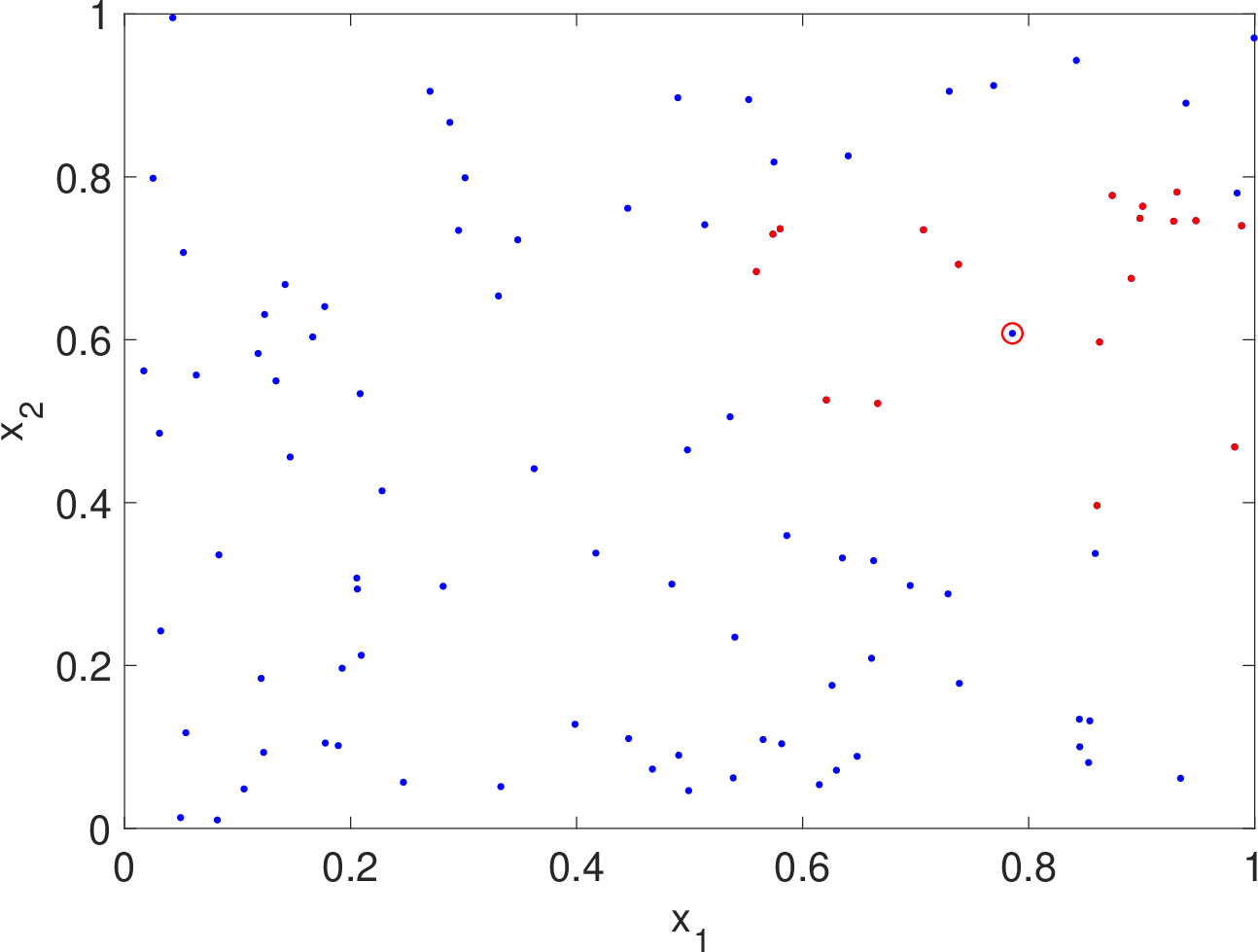}
    \end{subfigure}
    \begin{subfigure}[b]{0.22\textwidth}
        \centering
        \includegraphics[width=0.95\textwidth]{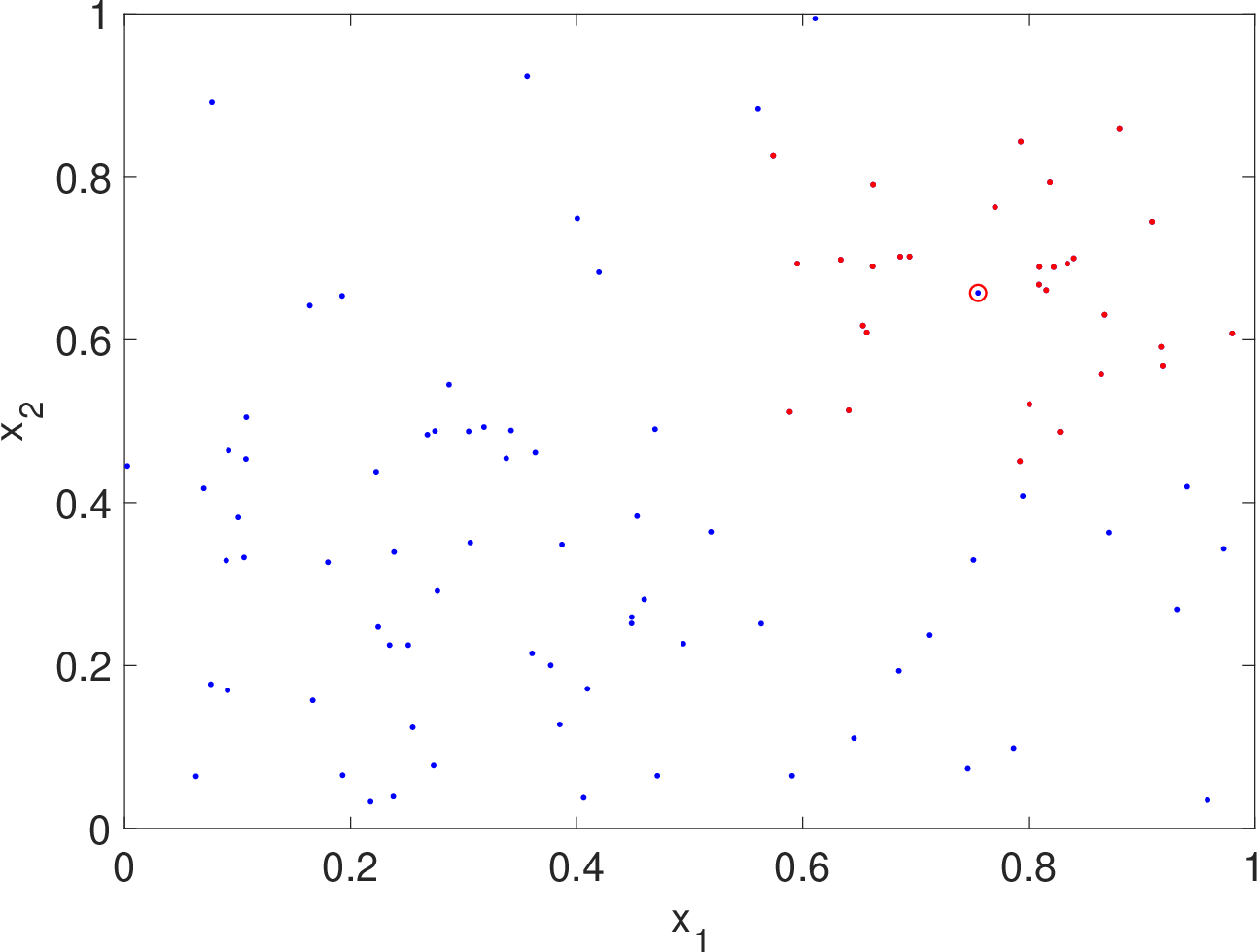}
    \end{subfigure}
    \begin{subfigure}[b]{0.22\textwidth}
        \centering
        \includegraphics[width=0.95\textwidth]{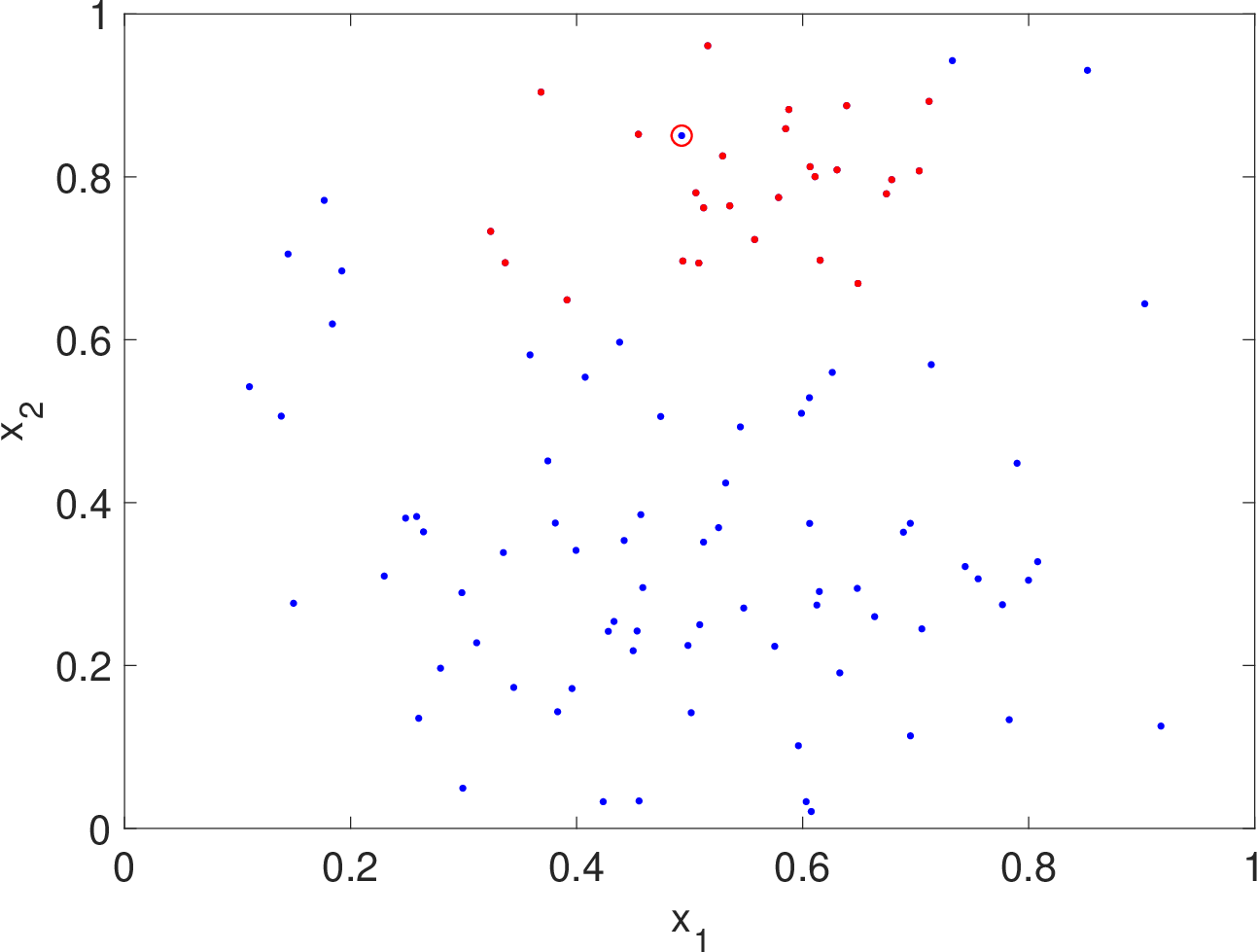}
    \end{subfigure}
    \begin{subfigure}[b]{0.22\textwidth}
        \centering
        \includegraphics[width=0.95\textwidth]{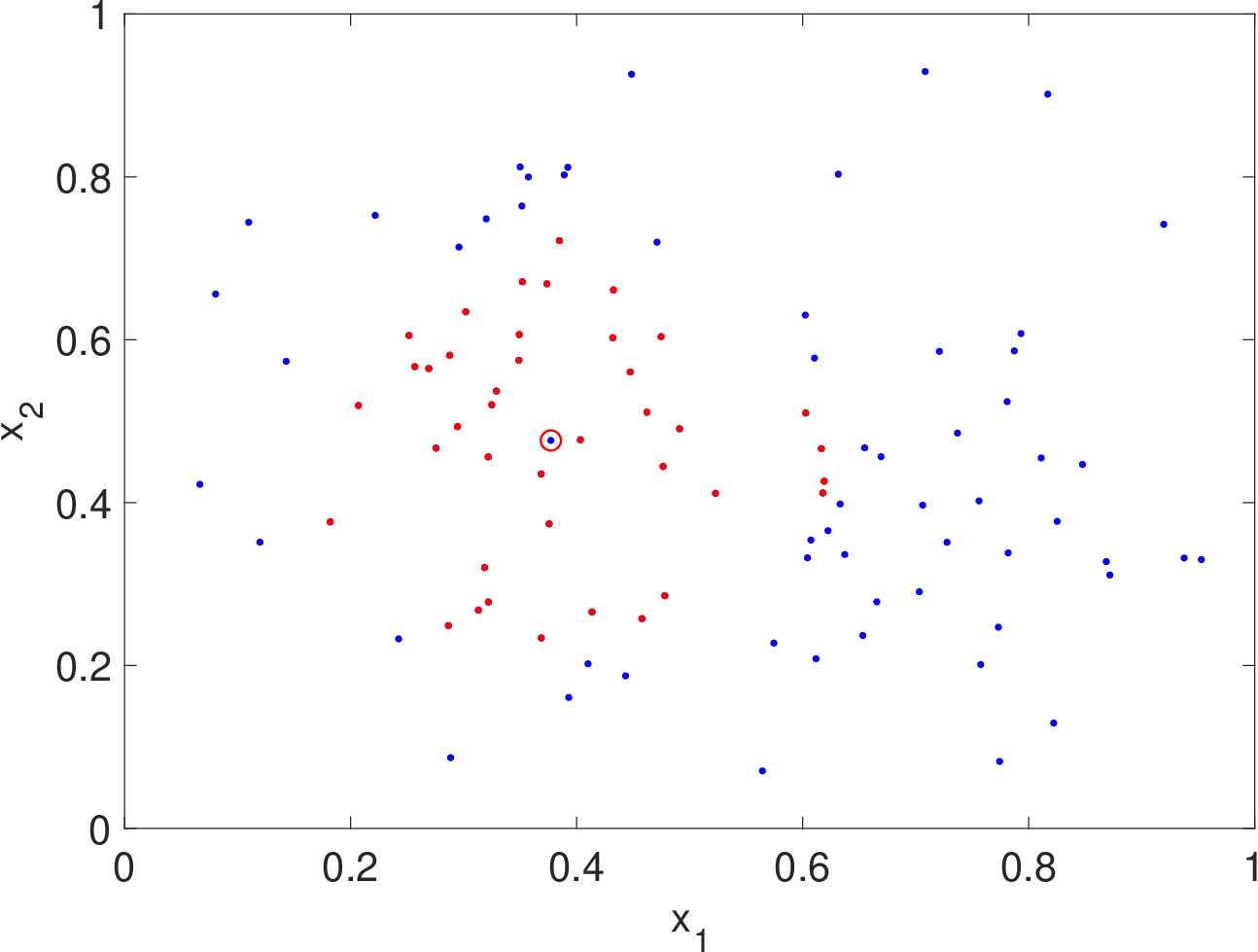}
    \end{subfigure}
    \vspace{3mm}
    
    \begin{subfigure}[b]{0.22\textwidth}
        \centering
        \includegraphics[width=0.95\textwidth]{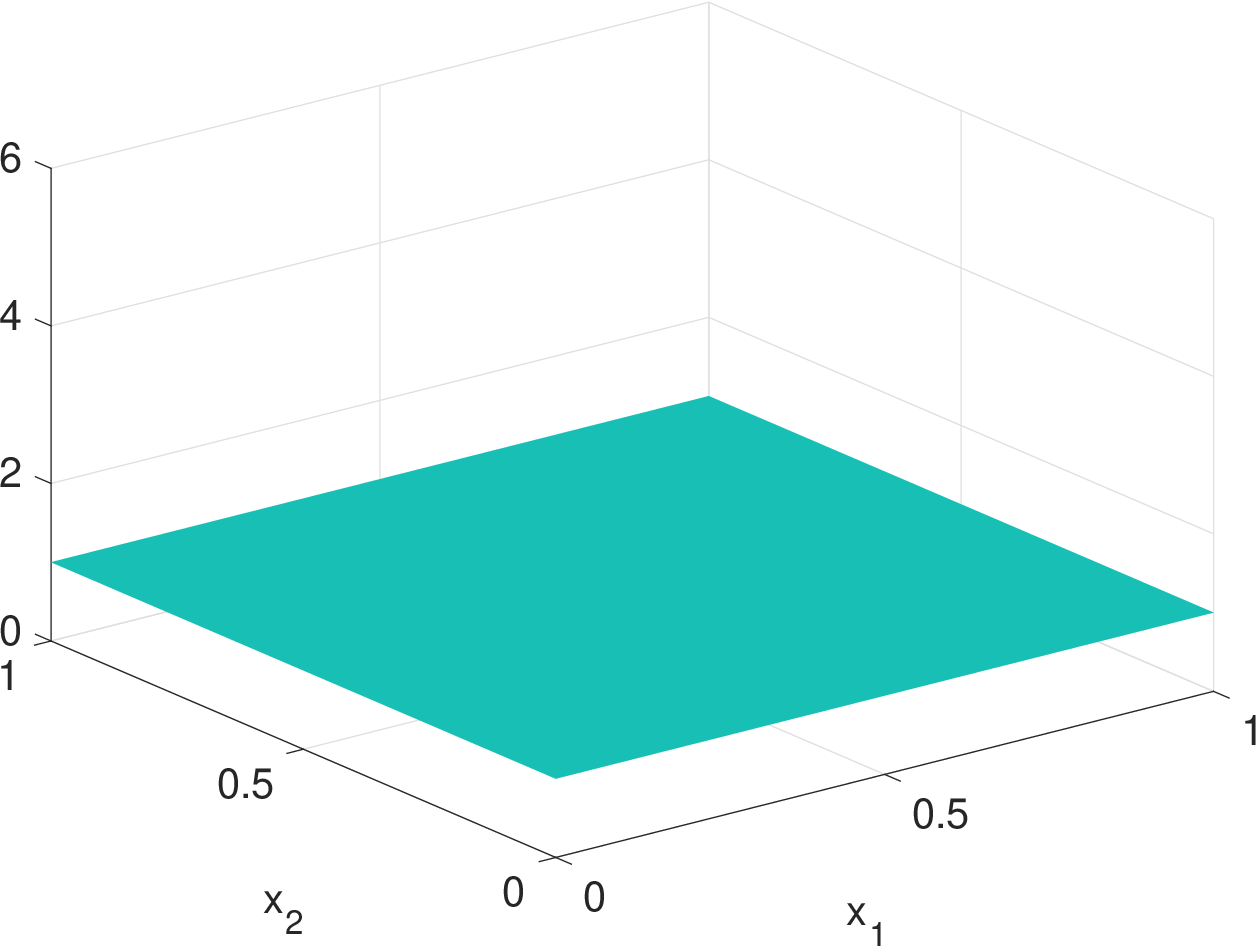}
    \end{subfigure}
    \begin{subfigure}[b]{0.22\textwidth}
        \centering
        \includegraphics[width=0.95\textwidth]{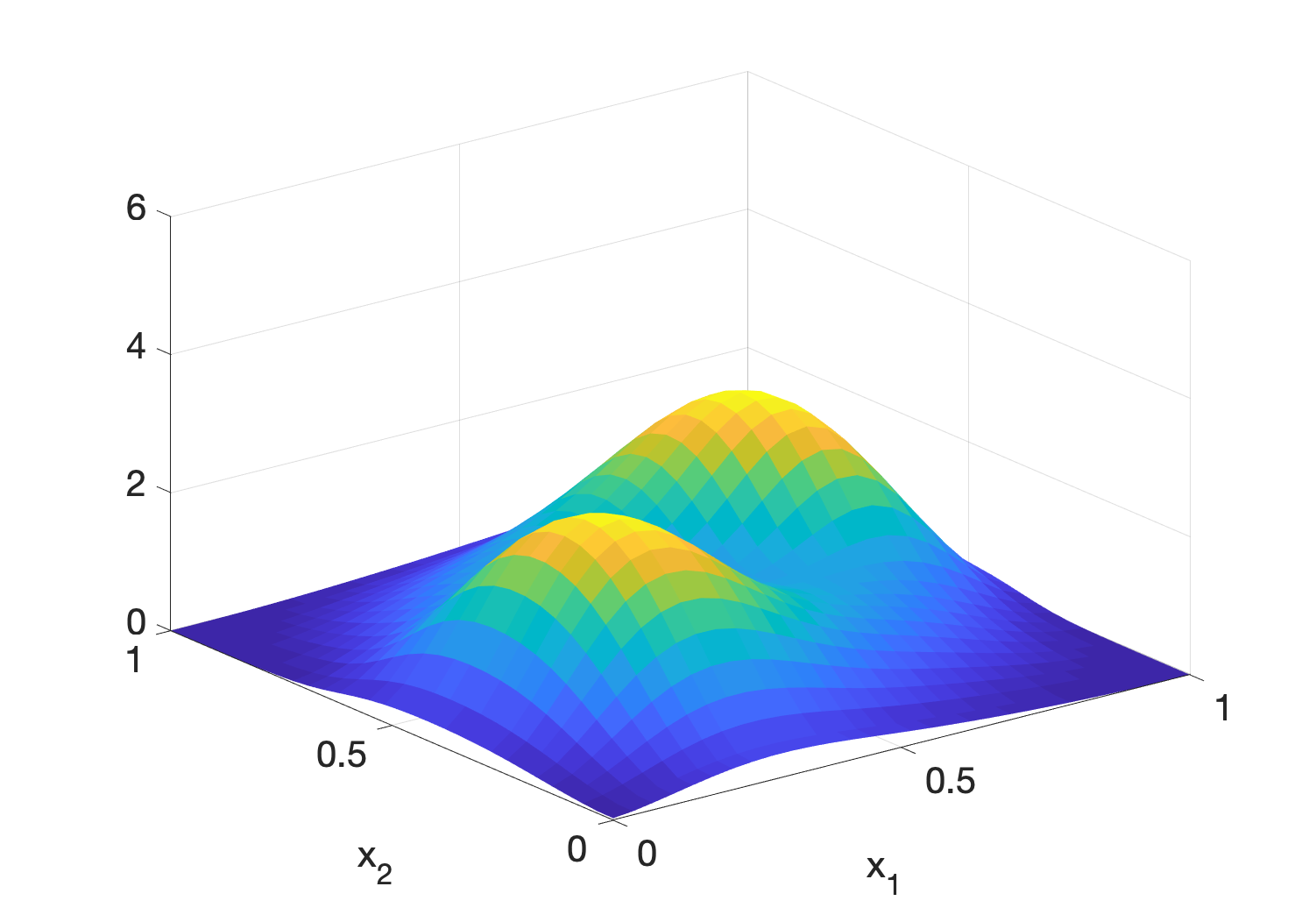}
    \end{subfigure}
    \begin{subfigure}[b]{0.22\textwidth}
        \centering
        \includegraphics[width=0.95\textwidth]{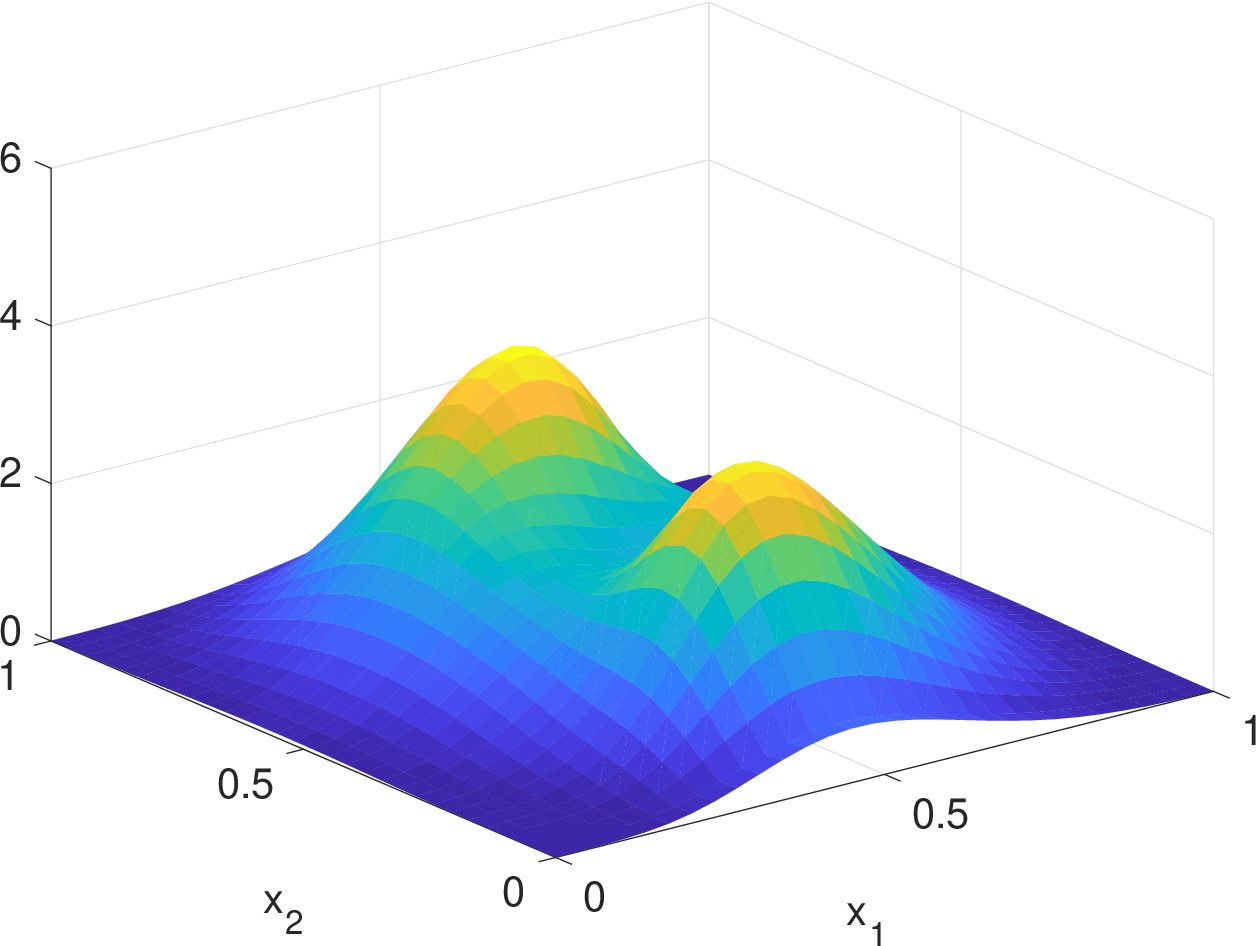}
    \end{subfigure}
    \begin{subfigure}[b]{0.22\textwidth}
        \centering
        \includegraphics[width=0.95\textwidth]{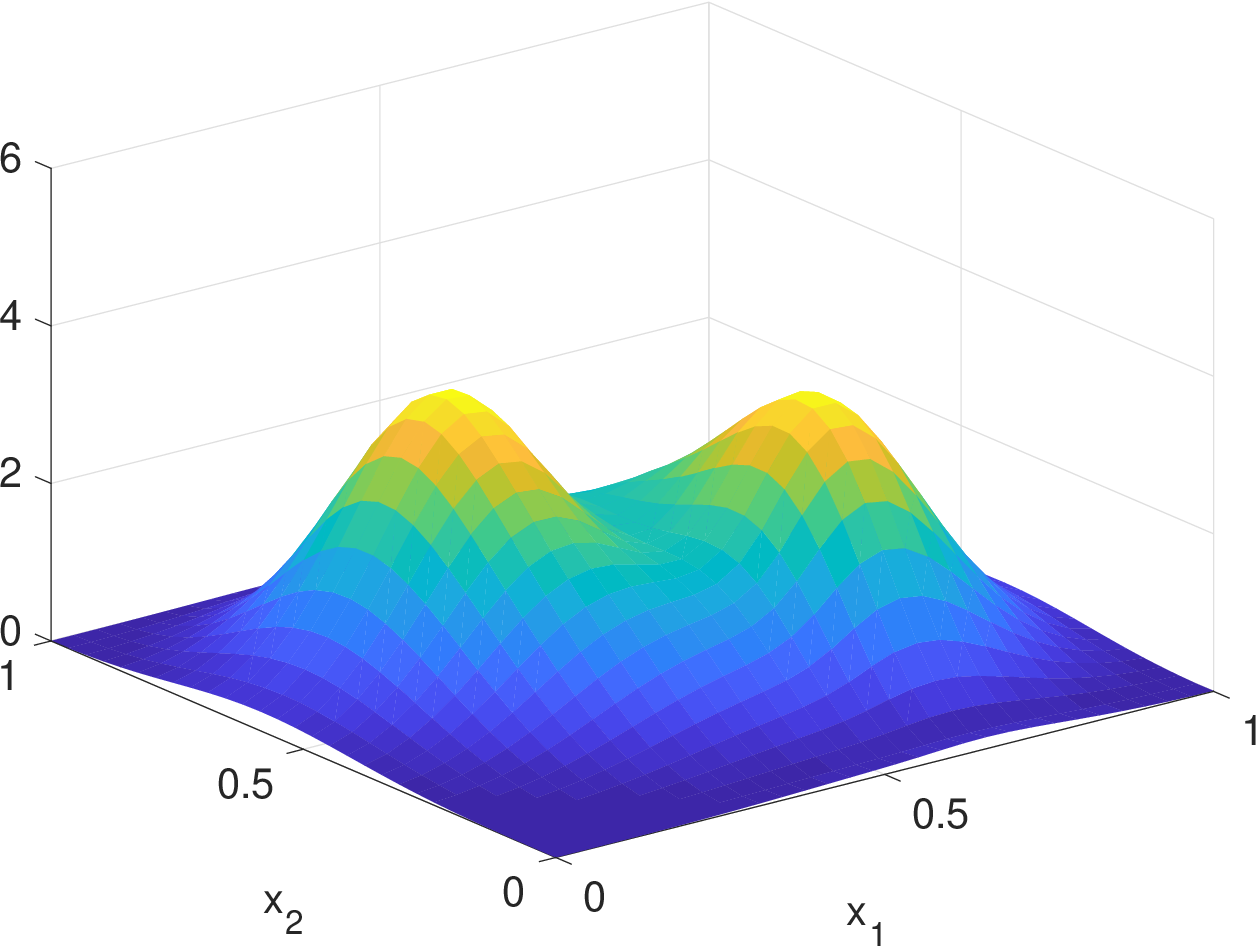}
    \end{subfigure}
    \vspace{3mm}


    \begin{subfigure}[b]{0.22\textwidth}
        \centering
        \includegraphics[width=0.95\textwidth]{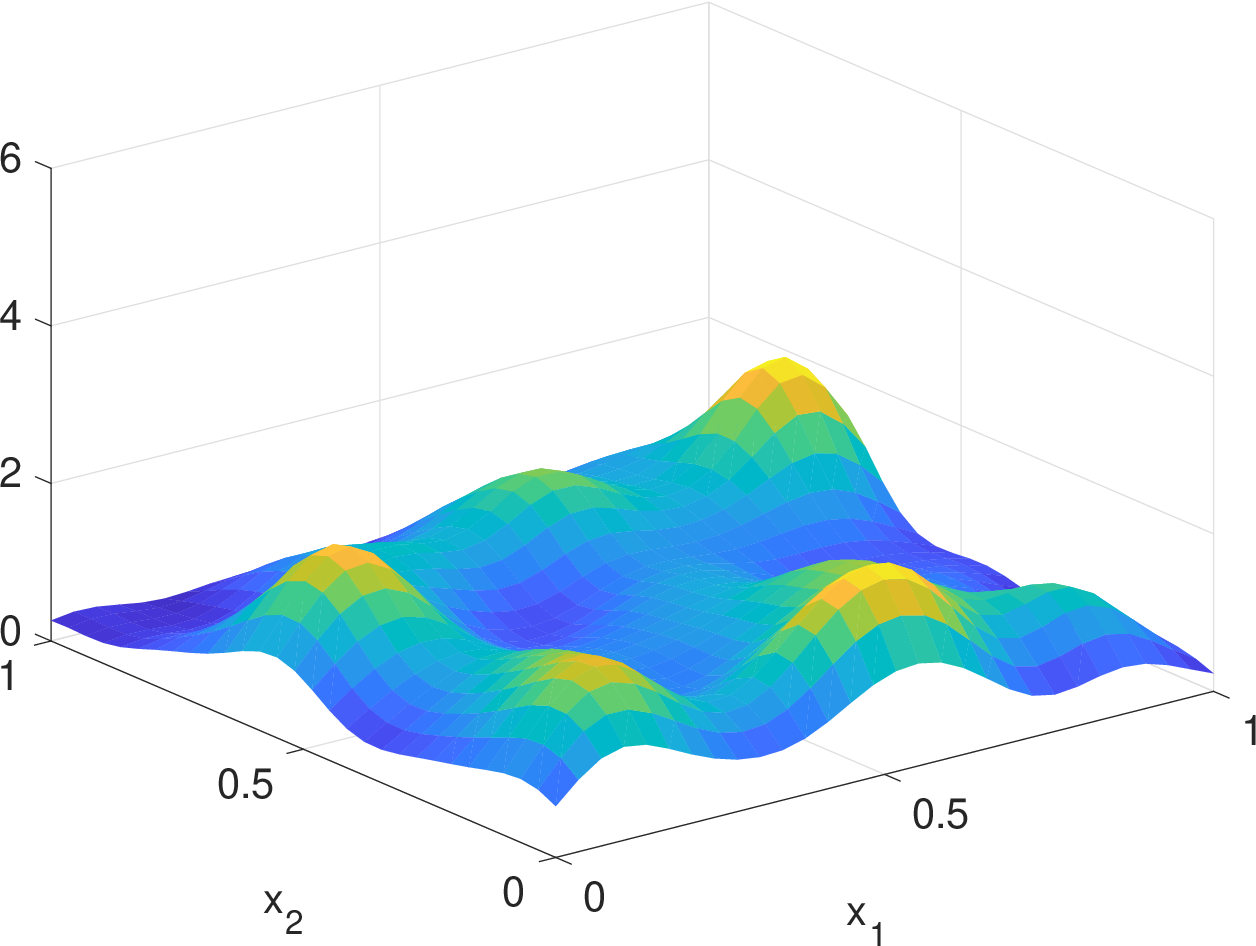}
    \end{subfigure}
    \begin{subfigure}[b]{0.22\textwidth}
        \centering
        \includegraphics[width=0.95\textwidth]{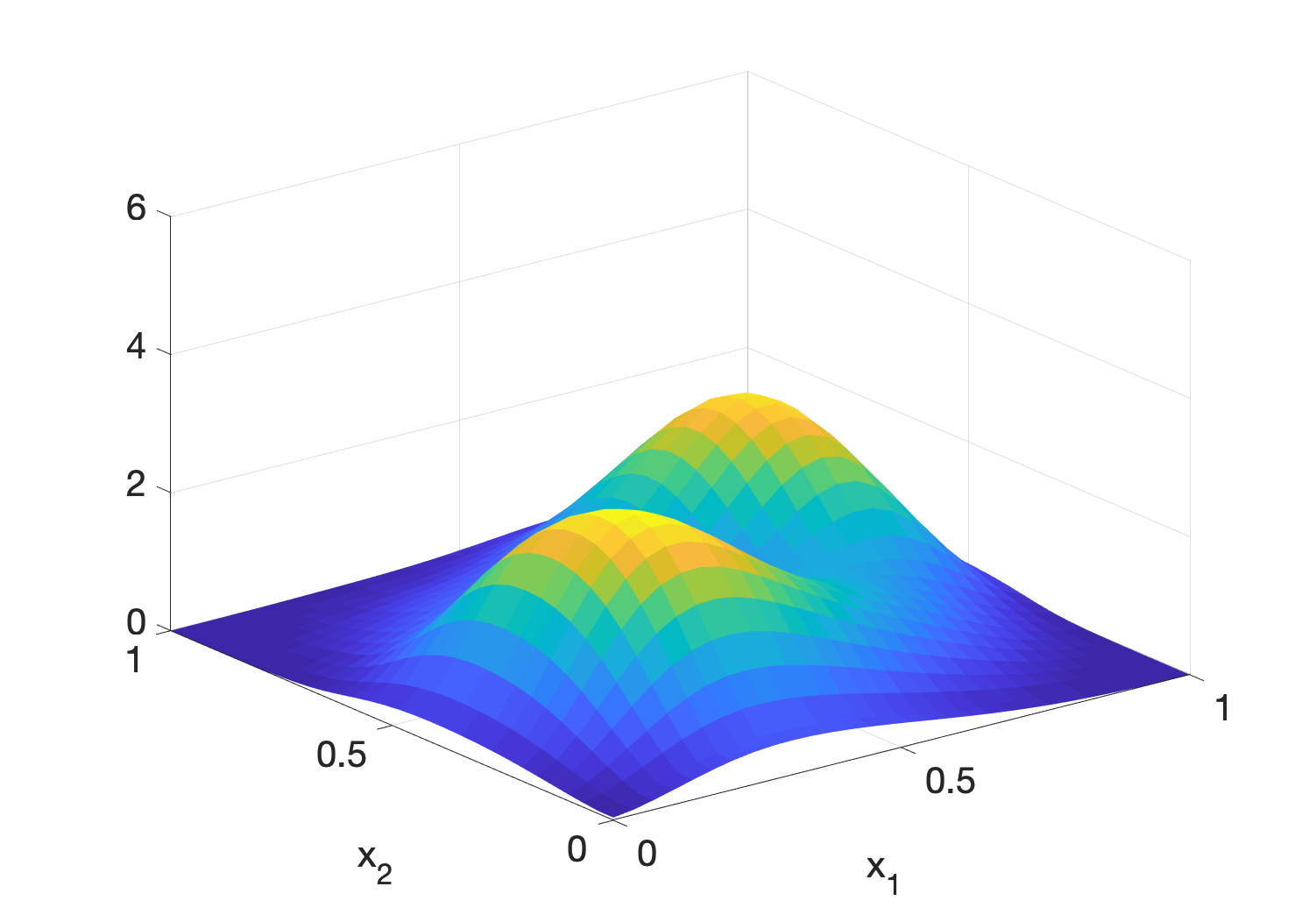}
    \end{subfigure}
    \begin{subfigure}[b]{0.22\textwidth}
        \centering
        \includegraphics[width=0.95\textwidth]{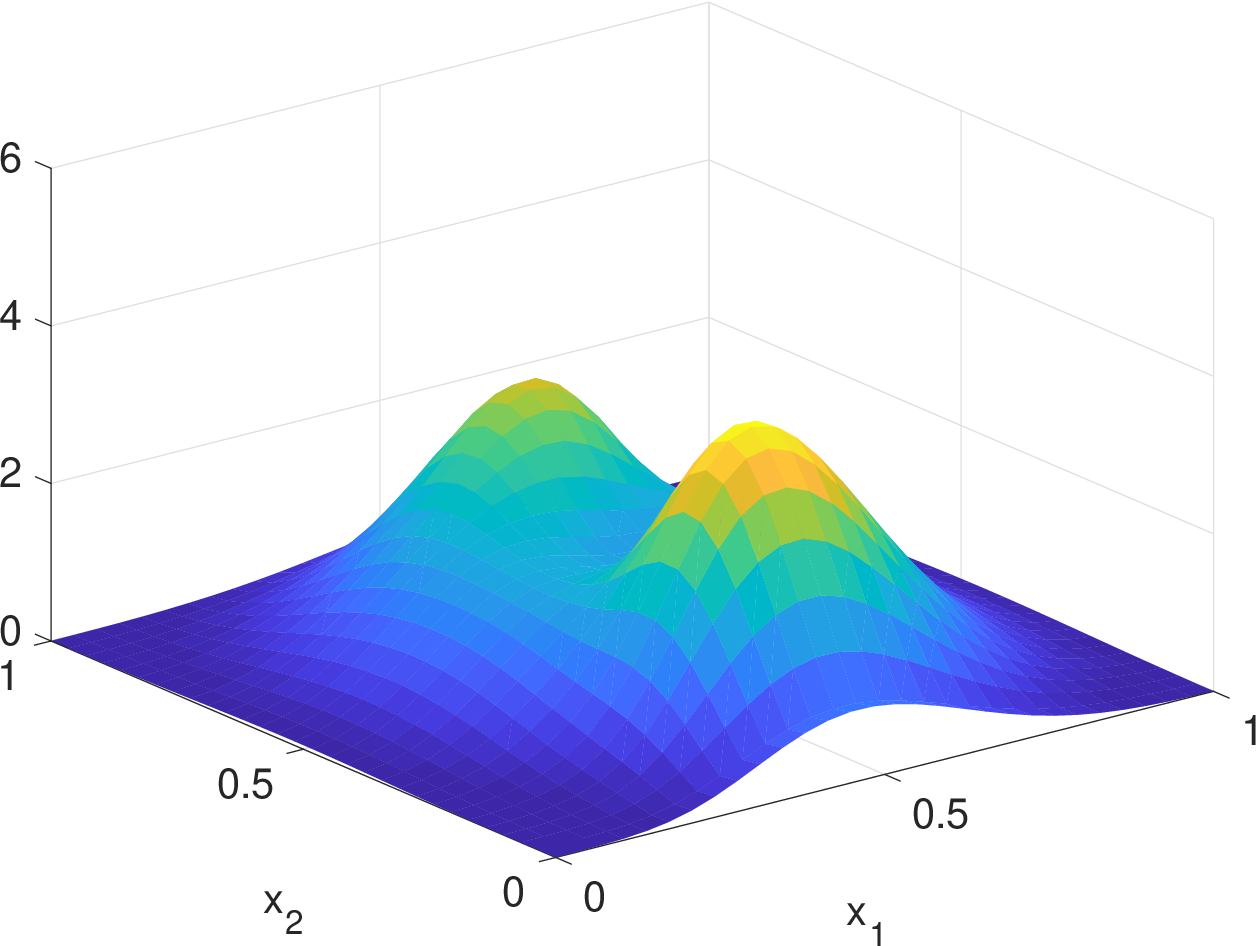}
    \end{subfigure}
    \begin{subfigure}[b]{0.22\textwidth}
        \centering
        \includegraphics[width=0.95\textwidth]{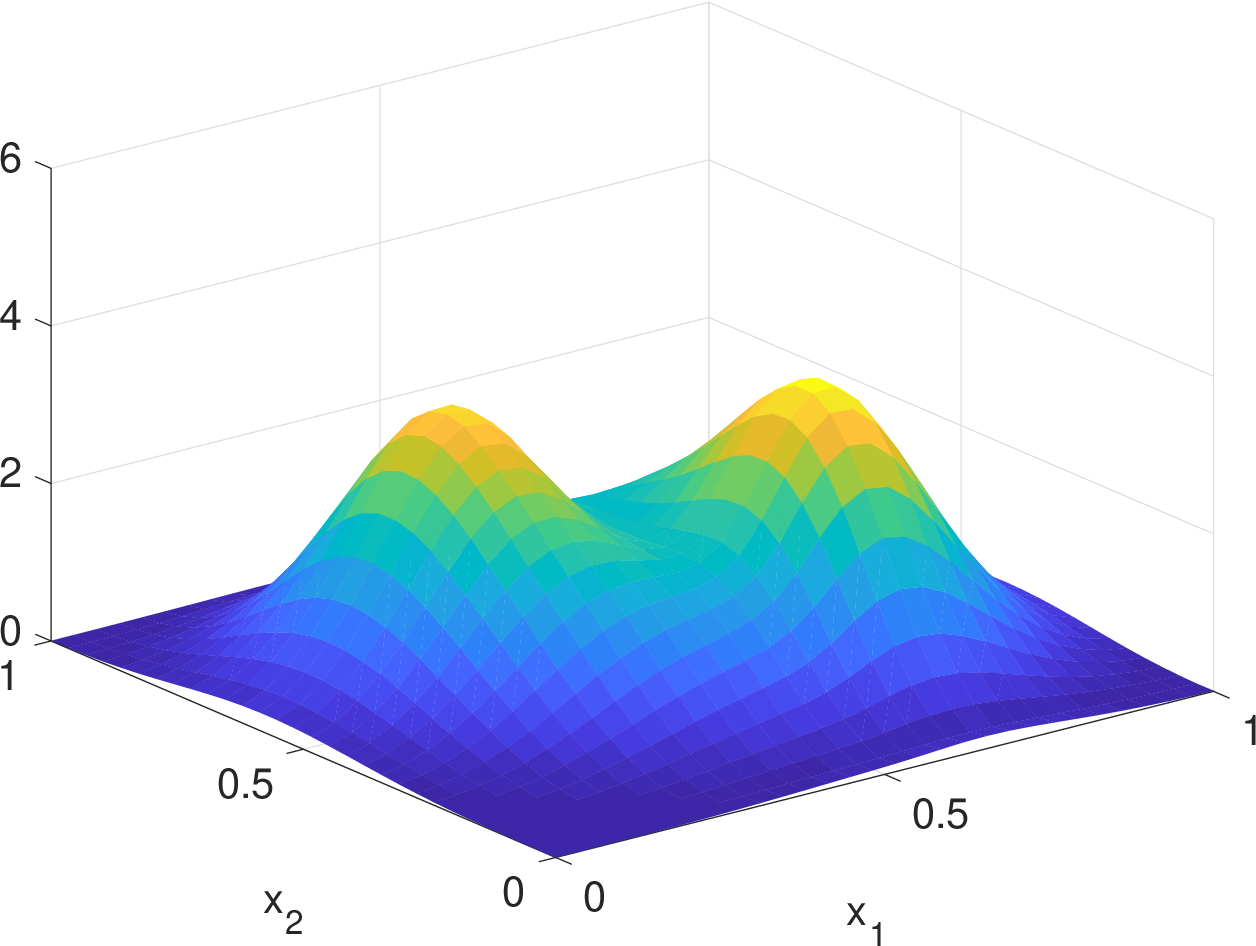}
    \end{subfigure}
    \vspace{3mm}
    
    \begin{subfigure}[b]{0.22\textwidth}
        \centering
        \includegraphics[width=0.95\textwidth]{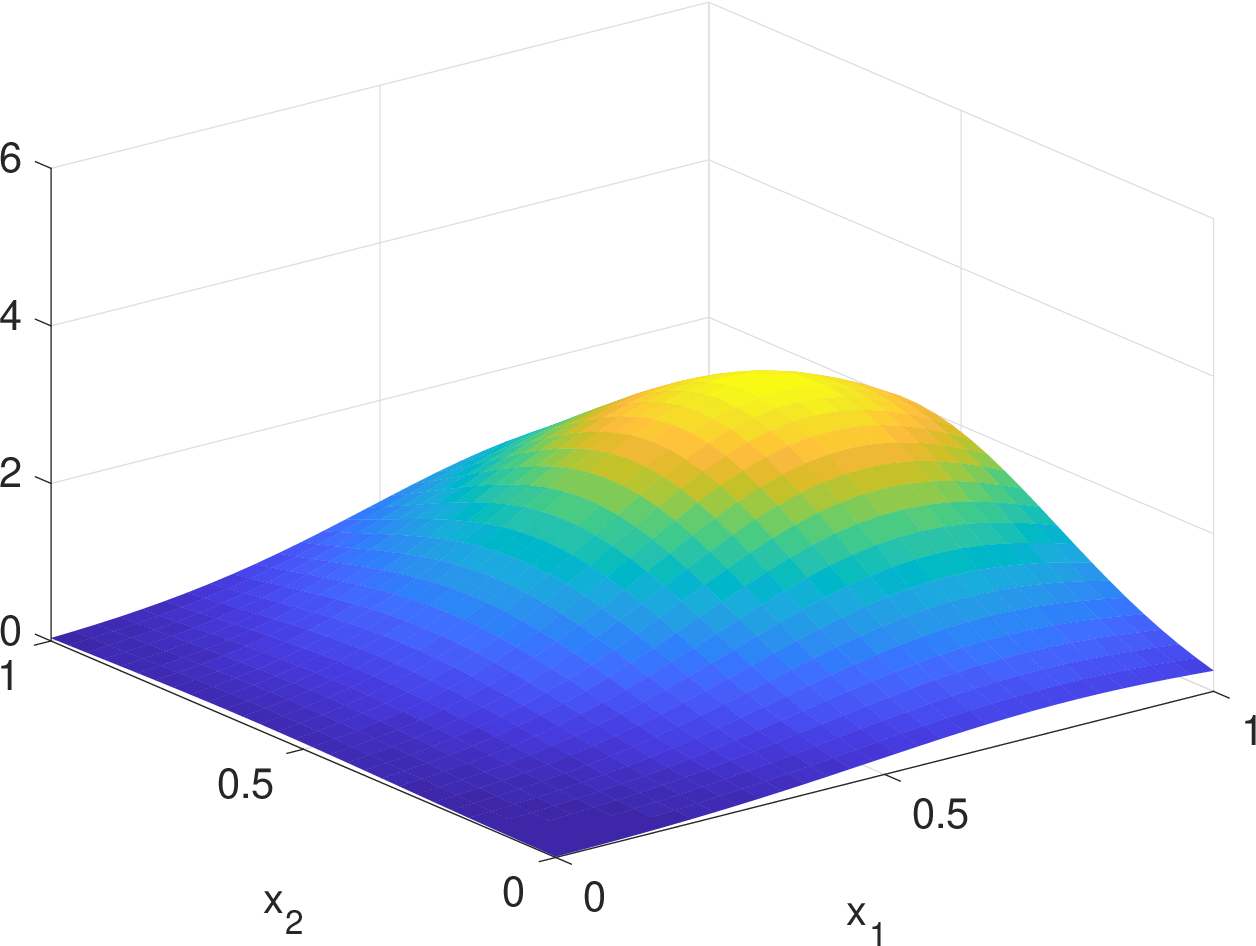}
    \end{subfigure}
    \begin{subfigure}[b]{0.22\textwidth}
        \centering
        \includegraphics[width=0.95\textwidth]{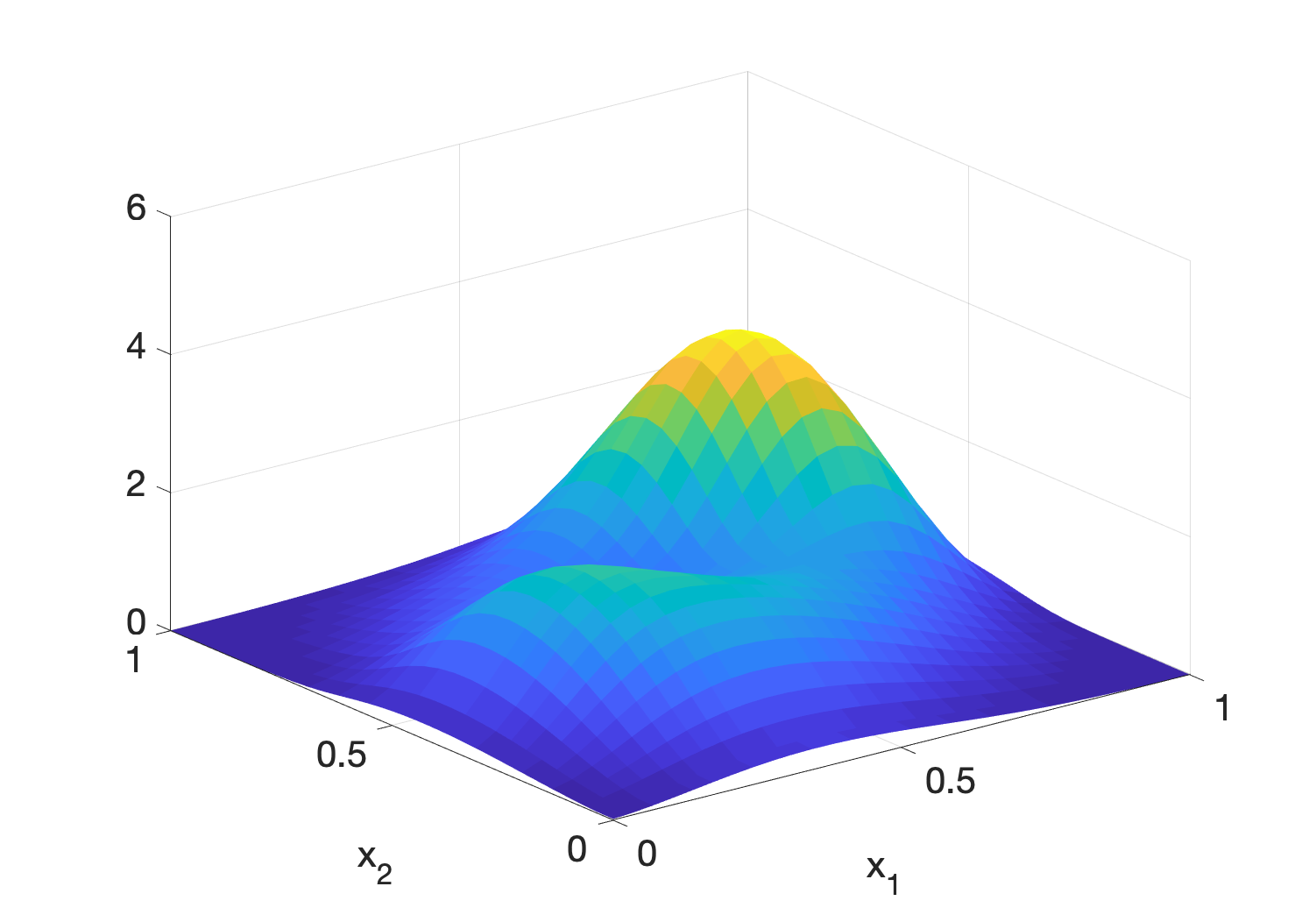}
    \end{subfigure}
    \begin{subfigure}[b]{0.22\textwidth}
        \centering
        \includegraphics[width=0.95\textwidth]{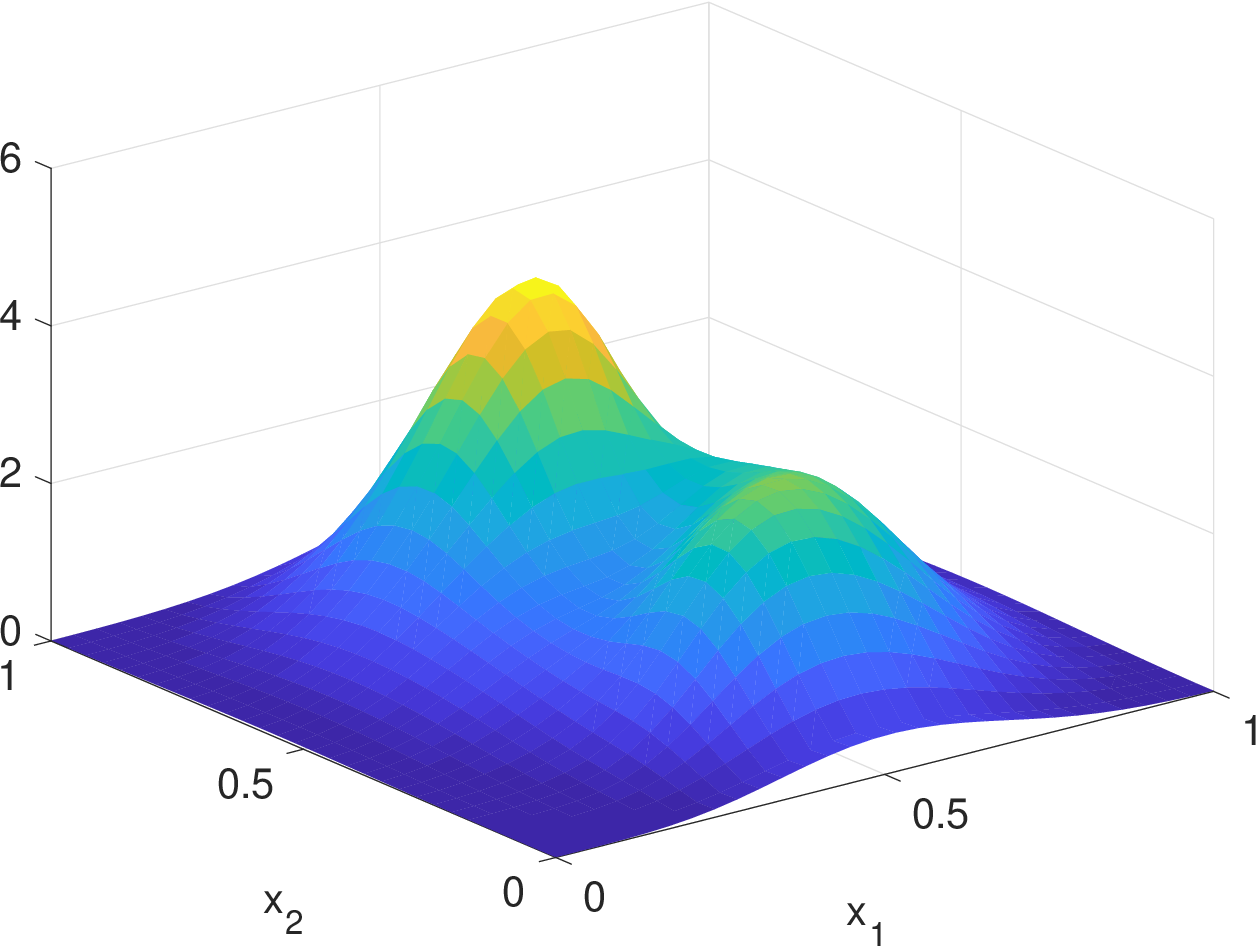}
    \end{subfigure}
    \begin{subfigure}[b]{0.22\textwidth}
        \centering
        \includegraphics[width=0.95\textwidth]{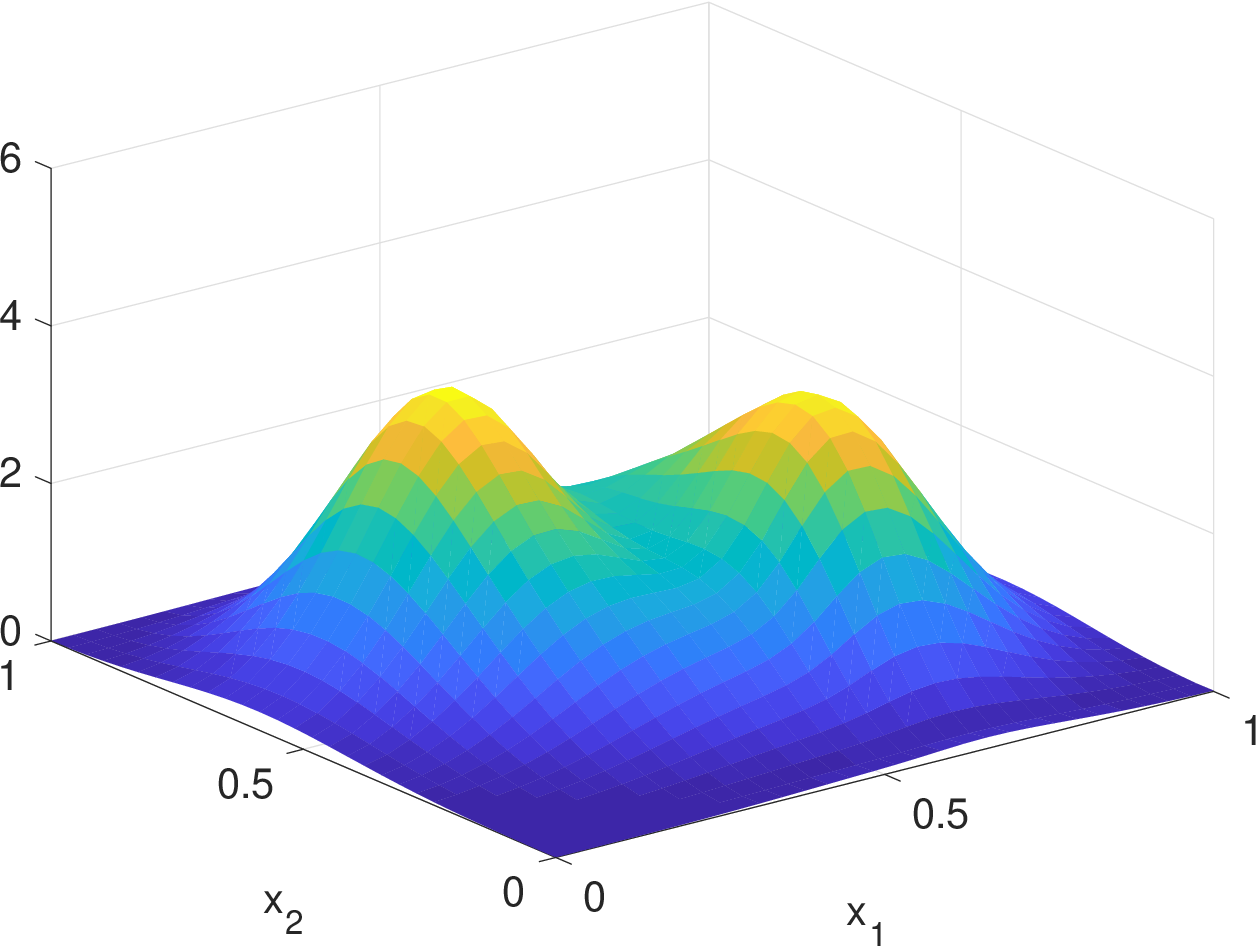}
    \end{subfigure}
    
    \caption{Agents' states where the red circle is the representative agent and the red dots are its neighbors (1st row), the ground truth density $p(x,t)$ (2nd row),
    outputs of centralized filter $\Hat{p}(x,t)$ (3rd row), and outputs of local filter $\Hat{p}_i(x,t)$ (4th row).}
    \label{fig:density filter}
\end{figure*}

\section{Conclusion}\label{section:conclusion}
We have presented a distributed density filter for estimating the dynamic density of large-scale systems with known dynamics by a novel integration of mean-fields models, KDE, infinite-dimensional Kalman filters and consensus protocols. 
With the distributed filter, each agent was able to estimate the global density using only the dynamics, its own state/position and state/position exchange with neighbors.
It was scalable to the population of agents, convergent in estimation error, and very efficient in communication and computation.
This algorithm can be used for many density-based distributed optimization and control problems of large-scale systems when density feedback information is required.
Our future work is to integrate the density filters into the mean-field feedback control framework we recently proposed to achieve fully distributed control of large-scale systems \cite{zheng2020transporting, zheng2021field}.

\section*{Appendices: dynamic average consensus}
We introduce the PI consensus estimator presented in \cite{freeman2006stability}.
Consider a group of $N$ agents where each agent has a local reference signal $u_{i}(t):[0, \infty) \to \mathbb{R}$.
The dynamic average consensus problem consists of designing an algorithm such that each agent tracks the time-varying average  $u_{\text{avg}}(t):=\frac{1}{N}\sum_{i=1}^{N}u_i(t)$.
The PI consensus algorithm is given by \cite{freeman2006stability}:
\begin{align}
    \dot{\nu}_{i} &=-\alpha\left(\nu_{i}-u_{i}\right)-\sum_{j=1}^{N} a_{i j}\left(\nu_{i}-\nu_{j}\right)+\sum_{j=1}^{N} b_{j i}\left(\eta_{i}-\eta_{j}\right)\nonumber \\
    \dot{\eta}_{i} &=-\sum_{j=1}^{N} b_{i j}\left(\nu_{i}-\nu_{j}\right)\label{eq:PI consensus estimator}
\end{align}
where $u_{i}$ is agent $i$'s reference input, $\eta_{i}$ is an internal state, $\nu_{i}$ is agent $i$'s estimate of $u_{\text{avg}}$, $[a_{i j}]_{N\times N}$ and $[b_{i j}]_{N\times N}$ are adjacency matrices of the communication graph, and $\alpha>0$ is a parameter determining how much new information enters the dynamic averaging process. 
The Laplacian matrices associated with $\left[a_{i j}\right]$ and $\left[b_{i j}\right]$ are represented by $L_{P}$ (proportional) and $L_{I}$ (integral) respectively. 
The PI estimator solves the consensus problem under constant (or slowly-varying) inputs, and remains stable for varying inputs in the sense of ISS.
Define the tracking error of agent $i$ by $\epsilon_{i}(t)=\nu_{i}(t)-u_{\text{avg}}(t),i=1,\dots,N$.
Decompose the error into the consensus direction $\mathbf{1}_{N}$ and the disagreement directions orthogonal to $\mathbf{1}_{N}$. 
Define the transformation matrix $\mathbf{T}=[(1 / \sqrt{N}) 1_{N} ~\mathbf{R}]$ where $\mathbf{R} \in \mathbb{R}^{N \times(N-1)}$ is such that $\mathbf{T}^{\top} \mathbf{T}=\mathbf{T} \mathbf{T}^{\top}=\mathbf{I}_{N},$ and consider the change of variables
$$
\overline{\epsilon}=\left[\begin{array}{c}
\bar{\epsilon}_{1} \\
\overline{\epsilon}_{2: N}
\end{array}\right]=\mathbf{T}^{\top}\epsilon, 
\quad 
\mathbf{w}=\left[\begin{array}{c}
w_{1} \\
\mathbf{w}_{2:N}
\end{array}\right]=\mathbf{T}^{\top} \mathbf{\eta}.
$$
The stability result is given as follows.

\begin{lemma}\label{lemma:ISS of PI estimator}
\cite{kia2019tutorial} Let $L_{\mathrm{I}}$ and $L_{P}$ be Laplacian matrices of strongly connected and weight-balanced digraphs. 
Then
\begin{align*}
\begin{split}
    \left|\epsilon_{i}(t)\right| \leq & \kappa e^{-\lambda\left(t-t_{0}\right)}\left\|\left[\begin{array}{c}
    \mathbf{w}_{2: N}\left(t_{0}\right) \\
    \overline{\epsilon}\left(t_{0}\right)
    \end{array}\right]\right\| \\
    &+\frac{\kappa\gamma}{\lambda} \sup _{t_{0} \leq \tau \leq t}\left\|\left(\mathbf{I}_{N}-\frac{1}{N} \mathbf{1}_{N} \mathbf{1}_{N}^{\top}\right) \dot{\mathbf{u}}(\tau)\right\|
\end{split}
\end{align*}
where $\kappa, \lambda, \gamma>0$ are constants depending on $\alpha$, ${L}_P$ and ${L}_I$.
\end{lemma}

For switching networks, each switch introduces a transient to the estimator error, and $\epsilon_i$ is still ISS \cite{lynch2008decentralized}.
The PI consensus estimator is also robust to initialization errors and permanent agent dropout in the sense that \cite{kia2019tutorial}:
\begin{equation}\label{eq:robust to initialization}
    \sum_{i=1}^{N} \nu_{i}(t)=\sum_{i=1}^{N} u_{i}(t)+e^{-\alpha\left(t-t_{0}\right)}\left(\sum_{i=1}^{N} \nu_{i}(t_{0})-\sum_{i=1}^{N} u_{i}(t_{0})\right).
\end{equation}

\bibliographystyle{IEEEtran}
\bibliography{References}

\end{document}